\documentclass[reprint, superscriptaddress,longbibliography]{revtex4-2}

\usepackage[T1]{fontenc}
\usepackage[utf8]{inputenc}

\usepackage{float}
\usepackage{graphicx}
\usepackage{epsfig,epstopdf}
\usepackage[table]{xcolor}

\usepackage[english]{babel}

\PassOptionsToPackage{hyphens}{url}
\usepackage[hyphenbreaks]{breakurl}

\usepackage{amsmath}
\usepackage{amsfonts}
\usepackage{amssymb,amsthm}
\usepackage{bbm}
\usepackage{mathtools}
\usepackage{tikz}
\usetikzlibrary{backgrounds,intersections, calc,through} 
\usepackage{quantikz}
\usepackage{relsize}

\usepackage{physics}
\usepackage{xprintlen}

\usepackage{complexity}

\usepackage{enumerate}

\PassOptionsToPackage{pdftex,hyperfootnotes=false,pdfpagelabels}{hyperref}
\usepackage{hyperref}
\hypersetup{
    colorlinks=true, linktocpage=true, pdfstartpage=3, pdfstartview=FitV,
    breaklinks=true, pdfpagemode=UseNone, pageanchor=true, pdfpagemode=UseOutlines,
    plainpages=false, bookmarksnumbered, bookmarksopen=true, bookmarksopenlevel=1,
    hypertexnames=true, pdfhighlight=/O,
    urlcolor=blue!60!black, linkcolor=blue!60!black, citecolor=gray, 
}   

\usepackage[capitalise,compress]{cleveref}

\makeatletter 
\hypersetup{
      pdftitle = {
            On spoofing deep thermalization
            },
      pdfauthor = {
            Shantanav Chakraborty, 
            Soumik Ghosh, 
            Tudor Giurgica-Tiron 
            },
        }
\makeatother

\newtheorem{theorem}{Theorem}
\newtheorem{definition}[theorem]{Definition}
\newtheorem{lemma}[theorem]{Lemma}

\newtheorem*{remark}{Remark}

\renewcommand{\l}[0]{\langle}

\newclass{\np}{NP}
\newclass{\sP}{\#P}
\newclass{\postBQP}{postBQP}
\newclass{\sampBPP}{sampBPP}

\usepackage[all]{hypcap}
\usepackage{url}
\usepackage{textcomp, complexity}
\usepackage{tikz}
\usetikzlibrary{decorations.pathreplacing}
\usepackage{physics}
\usepackage{tabularx,environ}
\usepackage{multirow}
\usepackage{mathpazo}
\usepackage{mathrsfs}

\let\oldnl\nl
\newcommand{\nonl}{\renewcommand{\nl}{\let\nl\oldnl}}
\newcommand{\ccol}[2]{\color{#1}#2\color{black}}

\usepackage[capitalise,compress]{cleveref}
\crefname{section}{Section}{Section} 
\crefname{subsection}{Section}{Section}

\newtheoremstyle{mystyle}
{3pt}
{3pt}
{\upshape}
{}
{\bfseries}
{.}
{.5em}
{}

\theoremstyle{theorem}
\crefname{thm}{Theorem}{Theorems}
\Crefname{thm}{Theorem}{Theorems}

\crefname{corll}{Corollary}{Corollaries}
\theoremstyle{remark}

\theoremstyle{definition}

\theoremstyle{mystyle}
\newtheorem{todoaux}{To-Do}

\NewDocumentEnvironment{todo}{o}
 {\IfNoValueTF{#1}
   {\todoaux\addcontentsline{toc}{subsection}{\protect\numberline{\thesubsection}\ccol{red}{To-Do item}}}
   {\todoaux[#1]\addcontentsline{toc}{subsection}{\protect\numberline{\thesubsection}{\ccol{red}{To-Do item}} (#1)}}
   \ignorespaces}
 {\endtodoaux}

\makeatletter

\newcolumntype{\expand}{}
\long\@namedef{NC@rewrite@\string\expand}{\expandafter\NC@find}

\usepackage{environ}

\NewEnviron{task}[2][]{
  \def\problem@arg{#1}
  \def\problem@framed{framed}
  \def\problem@lined{lined}
  \def\problem@doublelined{doublelined}
  \ifx\problem@arg\@empty
    \def\problem@hline{}
  \else
    \ifx\problem@arg\problem@doublelined
      \def\problem@hline{\hline\hline}
    \else
      \def\problem@hline{\hline}
    \fi
  \fi
  \ifx\problem@arg\problem@framed
    \def\problem@tablelayout{|>{\bfseries}lX|c}
    \def\problem@title{\multicolumn{2}{|l|}{
        \raisebox{-\fboxsep}{\textsc{#2}}
      }}
  \else
    \def\problem@tablelayout{>{\bfseries}lXc}
    \def\problem@title{\multicolumn{2}{l}{
        \raisebox{-\fboxsep}{\textsc{#2}}
      }}
  \fi
\par
\noindent
  
  \begin{tabularx}{\columnwidth}{\expand\problem@tablelayout}
    \problem@hline
    \problem@title\\[\fboxsep]
    \BODY 
  \end{tabularx}
    \\
}

\usepackage{amsthm}
\newtheorem{fact}{Fact}

\newtheorem{construction}{Construction}[section]

\theoremstyle{remark}

\theoremstyle{definition}
\newtheorem*{observation}{Observation}

\renewenvironment{proof}{\emph{\bfseries Proof.}}{\qed\\\vspace{10pt}}

\def\C{{\mathbb{C}}} 

\renewcommand{\Pr}{\mathop{\mathbb{P}\/}}
\renewcommand{\E}{\mathop{\mathbb{E}\/}}

\renewcommand{\abs}[1]{\left|#1\right|}

\newcommand{\Ber}{{\mathrm{Ber}}} 

\NewCommandCopy{\dashl}{\l} 
\renewcommand{\l}{\ell} 

\renewcommand{\tilde}{\widetilde}

\newcommand{\B}{\{0,1\}}

\newcommand{\calD}{\mathcal{D}}

\newcommand{\calF}{\mathcal{F}}

\newcommand{\calH}{\mathcal{H}}

\newcommand{\calP}{\mathcal{P}}

\newlang{\StateHSP}{StateHSP}
\newlang{\HSP}{HSP}
\newlang{\negl}{negl}

\newcommand{\genstirlingII}[3]{%
  \genfrac{\{}{\}}{0pt}{#1}{#2}{#3}%
}

\newcommand{\stirlingII}[2]{\genstirlingII{}{#1}{#2}}

\newcommand{\Sym}[2]{\mathrm{Sym}^{#1}\left(#2\right)}

\newcommand\powerset[1]{\mathcal{P}\left(#1\right)}
\newcommand{\Haar}[1]{\mathop{\text{Haar}\/}\left(#1\right)}

\newcommand{\PRP}{\mathsf{PRP}}
\newcommand{\PRF}{\mathsf{PRF}}

\begin{document}

\title{Fast computational deep thermalization} 
\author{Shantanav Chakraborty}
\email{shchakra@iiit.ac.in}
\affiliation{CQST and CSTAR, International Institute of Information Technology Hyderabad, Telangana 500032, India}

\author{Soonwon Choi}
\email{soonwon@mit.edu}
\affiliation{Department of Physics, Massachusetts Institute of Technology, Cambridge, MA 02139, USA}

\author{Soumik Ghosh}
\email{soumikghosh@uchicago.edu}
\affiliation{Department of Computer Science, University of Chicago, Chicago, Illinois 60637, USA}

\author{Tudor Giurgic\u{a}-Tiron}
\email{tudorgt@umd.edu}
\affiliation{QuICS, University of Maryland, College Park, Maryland 20742, USA}

\begin{abstract}
Deep thermalization refers to the emergence of Haar-like randomness from  quantum systems upon partial measurements.
As a generalization of quantum thermalization, it is often associated with high complexity and entanglement.
Here, we introduce \emph{computational deep thermalization} and construct the fastest possible dynamics exhibiting it at infinite effective temperature.
Our circuit dynamics produce quantum states with low entanglement in polylogarithmic depth that are indistinguishable from Haar random states to any computationally bounded observer. Importantly, the observer is allowed to request many copies of the same residual state obtained from partial projective measurements on the state --- this condition is beyond the standard settings of quantum pseudorandomness, but natural for deep thermalization.
In cryptographic terms, these states are pseudorandom, pseudoentangled, and crucially, retain these properties under local measurements. 
Our results demonstrate a new form of computational thermalization, where thermal-like behavior arises from structured quantum states endowed with cryptographic properties, instead of from highly unstructured ensembles. The low resource complexity of preparing these states suggests scalable simulations of deep thermalization using quantum computers. Our work also motivates the study of computational quantum pseudorandomness beyond $\mathsf{BQP}$ observers.
\end{abstract}
\maketitle

Quantum thermalization is a universal phenomenon occurring in generic interacting systems and underlies the foundation of quantum statistical physics. 
Its key feature is the emergence of universal randomness; despite being in pure states with no entropy, quantum systems may locally appear as if they are in maximally entropic thermal states~\cite{deutsch1991quantum, srednicki1994chaos, rigol2008thermalization, nandkishore2015review}.

Recent experiments enabled by a high-fidelity Rydberg quantum simulator discovered an intriguing possibility: in fact, quantum thermalization might be a special consequence of a stronger form of universal phenomena~\cite{choi2023preparingrandomstates,cotler2023emergent}.
\textit{Deep thermalization} refers to the behavior of highly entangled pure states that remain Haar-like even under partial projective measurements~\cite{Ippoliti_2022, Ho_2022, cotler2023emergent, ippoliti2023dynamical, yu2025mixedstate}.
Specifically, when one subsystem is measured in some fixed local basis, the resulting ensemble of conditional states on the remaining qubits still looks maximal entropic, forming approximate $t$-designs \cite{renes2004symmetric, ambainis2007quantumtdesigns}. That is, $t$-moments of these projected ensembles are indistinguishable from $t$-moments of a Haar-random state, also referred to as a \textit{projected $t$-design}. 
In this viewpoint, quantum thermalization reduces to the special case $t=1$.
Better understanding deep thermalization is an important problem with immediate implications in the study of quantum statistical physics as well as practical applications in  benchmarking~\cite{choi2023preparingrandomstates}, certifying state preparation~\cite{mcginley2023shadow,huang2024certifying}, and generating random states~\cite{mok2025optimalconversionclassicalquantum,zhang2025holographic}.

While the conjectured behavior of deep thermalization has been explored through numerical simulations, rigorous theoretical results remain scarce. Existing analyses apply only in highly restricted regimes, such as when the bath (i.e.\, the measured sub-system) size is infinite-dimensional \cite{choi2023preparingrandomstates, Ho_2022, mok2025optimalconversionclassicalquantum}, or for fixed, low values of $t$ \cite{cotler2023emergent}. Meanwhile, both numerical and experimental observations suggest that projected designs emerge far more quickly (at shallower depth circuits) than what current mathematical proofs account for~\cite{cotler2023emergent,choi2023preparingrandomstates,Ippoliti_2022}. 
These observations motivate a fundamental open question: \emph{What is the fastest route to generating quantum states that appear deeply thermal to any physically reasonable observer?}

We leverage quantum complexity theory to provide new perspectives on this question. The key insight is that, in realistic settings, observers do not have unrestricted access to all properties of a quantum state, but are instead constrained by the limits of efficient computation. A computationally bounded observer is one that can perform only polynomial-time quantum computations. In this context, two ensembles of quantum states are said to be \emph{computationally indistinguishable} if no such observer can efficiently tell them apart, even though they may be statistically distinct. Often, it turns out that these limitations help us in mimicking Haar randomness very fast \cite{JiLiuSong2018}, with very low quantities of entanglement and other resources \cite{aaronson_et_al:LIPIcs.ITCS.2024.2, giurgicatiron2023pseudorandomnesssubsetstates,jeronimo2024pseudorandompseudoentangledstatessubset,bouland_et_al:LIPIcs.CCC.2024.21, Gu_2024,gu2024simulatingquantumchaoschaos}. This raises a compelling possibility: perhaps the fastest route to prepare deeply thermalized states is not via chaotic dynamics, maximal entropy production, or exact unitary designs, but rather through computational pseudorandomness. However, it remains open whether existing constructions of pseudorandom states achieve the stronger notion of quantum randomness required for deep thermalization.

In this work, we construct an ensemble of $n$-qubit states using only polylogarithmic-depth circuits and minimal entanglement, which are nevertheless \textit{computationally indistinguishable} from deeply thermalized states at infinite temperatures. Specifically, our construction satisfies the following three properties:
\begin{itemize}
\item[(i)] The states are \textit{computationally indistinguishable} from Haar random states, upto $t$ copies, for \emph{any} $t=\poly(n)$ (arbitrary polynomial in $n$). 
\item[(ii)] The states need $1$D geometrically local circuits of depth $\sim \log^2n$. The states have low entanglement across geometric cuts, of the order $\sim \log^2n$.
\item[(iii)] Across almost all bipartitions of the system, the projected ensemble, obtained by measuring a part of the system in the computational basis, is \emph{also} computationally indistinguishable from Haar random states upto $t = \poly(n)$ copies and have a similar entanglement profile to that of the pre-measurement state. 
\end{itemize}

We note that these conditions are chosen to be as stringent as possible. Condition (i) ensures that the global state looks (pseudo-)thermal at infinite temperature.
The logarithmic circuit depth and low entanglement requirement in Condition (ii) saturate the known lower bound for pseudorandom states~\cite{JiLiuSong2018}.
Condition (iii) describes an additional requirement specific to deep thermalization: the projected ensemble must also appear pseudorandom while retaining low entanglement. Crucially, this holds even when the bath (i.e., the measured subsystem) is finite-sized (as small as $\log^2 n$ qubits). Therefore, Conditions (i) - (iii) combined together ask for quantum circuits that make the most drastic departure from all known analytical models of deep thermalization \cite{cotler2023emergent,choi2023preparingrandomstates,Ippoliti_2022, Ho_2022, ippoliti2023dynamical}, and instead aligns more closely with what has been experimentally and numerically observed \cite{choi2023preparingrandomstates,cotler2023emergent}. Indeed, existing solvable models for deep thermalization or preparing pseudorandom states only partially satisfy these conditions, but to the best of our knowledge, no known construction satisfies all three conditions ~\footnote{In 1D, the requirements in Conditions (ii) and (iii) are optimal due to geometric constraints. In higher dimensions, e.g., 2D, Condition (iii) can be achieved in constant depth via measurement-based computation~\cite{Briegel_2009}, but such states violate (i). Pseudorandom phase states~\cite{JiLiuSong2018} satisfy Conditions (i) and (iii), but exhibit high entanglement, violating Condition (ii). Finally, Subset states (or subset phase states)~\cite{aaronson_et_al:LIPIcs.ITCS.2024.2,jeronimo2024pseudorandompseudoentangledstatessubset} satisfy Conditions (i) and (ii), but not (iii), as shown later.}.
~\\~\\
\textbf{Deep thermalization:~}Consider a composite quantum system in an $N$-dimensional Hilbert space $\mathcal{H}_{AB} = \mathcal{H}_A \otimes \mathcal{H}_B$, prepared in a pure state $\ket{\psi}_{AB}$, where subsystems $A$ and $B$ have dimensions $N_A$ and $N_B$, respectively. The phenomenon of \emph{deep thermalization} refers to the emergence of approximate \emph{projected designs} in subsystem $A$ upon performing a projective measurement on subsystem $B$ in a fixed local basis (e.g., the computational basis $\{ \ket{z_k} \}_{k \in \{0,1\}^{N_B}}$). 

For instance, conditioned on outcome $k$, the post-measurement state on subsystem $A$ is given by
\[
\ket{\psi_A^{(k)}} := \frac{1}{\sqrt{p_k}} \left(I\otimes \ket{z_k}\bra{z_k}_B\right) \ket{\psi}_{AB},
\]
where $p_k = \left\| (I \otimes \ket{z_k}\bra{z_k})\ket{\psi}_{AB} \right\|^2$. This defines an ensemble $\{ p_k, \ket{\psi_A^{(k)}} \}$ on subsystem $A$. The principle of deep thermalization posits that this ensemble behaves, in a precise sense, like an ensemble of Haar-random states. Specifically, the $t$-th moment operator of the projected ensemble on $A$ approximates that of a Haar-random pure state $\ket{\phi}$ on $\mathbb{C}^{N_A}$:
\begin{align*}
   \left\|\mathbb{E}_k\dyad{\psi^{(k)}_A}^{\otimes t}-\E_{\phi\sim\mathrm{Haar}(\C^{N_A})}\dyad{\phi}^{\otimes t}\right\|_1\leq \mathrm{negl}(n), 
\end{align*}
where $\mathrm{negl}(n)$ is a function that decays with the number of qubits $n$, and $\|.\|_1$ is the 1-norm. 
~\\~\\
\textbf{Our work in context:~}Deep thermalization has been primarily explored through numerical simulations, especially in chaotic many-body systems, while rigorous theoretical results exist only in restricted settings. In particular, Haar randomness of the projected ensemble has been proven in certain special models of shallow 1D circuits, but only in the asymptotic limit $N_B\rightarrow\infty$ \cite{Ho_2022, ippoliti2023dynamical}. For finite $N_B$, sufficient conditions for deep thermalization were established in \cite{cotler2023emergent}. The authors prove that if a global state $\ket{\psi}_{AB}$ is drawn from an approximate $t$-design, then the projected ensemble in $A$ forms an approximate $t'$-design, provided $t' < t/N_B$ and $N_B\gg N_A$. Thus, the projected ensemble inherits only ``lower order" designs, with $t'\ll t$. Moreover, the proof only works for a fixed $t$, and not for arbitrary polynomially large moments. The infusion of classical randomness can, in principle, increase $t'$, but only by a constant factor \cite{mok2025optimalconversionclassicalquantum}. 

In contrast, a Haar random state remains Haar-like, with high probability, under projective measurements, even when $N_B$ is finite. In this sense, $t$-designs do not serve as faithful models of deep thermalization beyond extremely low-order moments, underscoring the need for new frameworks to capture this phenomenon. Importantly, the existence of bounded-depth state ensembles exhibiting deep thermalization was posed as an open question in \cite{cotler2023emergent}, where preliminary numerical evidence suggested that shallow brickwork circuits might suffice. Our work resolves this question affirmatively: we provide the first rigorous construction showing that deep thermalization can emerge from ensembles generated by polylogarithmic-depth circuits, even when the underlying states have minimal entanglement and arise from highly structured, pseudorandom phase patterns.

Prior constructions of pseudorandom or pseudoentangled states also do not exhibit deep thermalization, a fact also observed in \cite{feng2025dynamics}. Consider the following $n$ qubit state, known as a subset phase state, defined as 
\begin{equation}
\label{phase states2}
\ket{\psi_{f,S}} = \frac{1}{\sqrt{M}} \sum_{x \in S} (-1)^{f(x)} \ket{x},
\end{equation}
where $f$ is a pseudorandom function and $S\subset\{0,1\}^n$ is a (pseudo-)randomly chosen subset of size $M$. These states are pseudorandom and have very low entanglement  (pseudoentangled) across any cut \cite{aaronson_et_al:LIPIcs.ITCS.2024.2}. However, the corresponding projected ensemble leads to trivial standard basis states and not a $t$-design, violating property (iii), as sketched in Sec.~\ref{section: subset states are not projected designs} of the Supplemental Material (SM) \cite{supmat}. 
The same argument also holds for the so-called subset states \cite{jeronimo2024pseudorandompseudoentangledstatessubset, giurgicatiron2023pseudorandomnesssubsetstates}, which are of the same form as Eq.~\eqref{phase states2}, with the phase function removed.  On the other hand, phase states can be obtained when $S=\{0,1\}^n$, i.e.\ $M=2^n$ \cite{JiLiuSong2018, aaronson_et_al:LIPIcs.ITCS.2024.2}. Measuring such a state across any bipartition of qubits yields another phase state that retains its pseudorandomness. However, phase states exhibit maximal entanglement both before and after measurement, thereby violating Conditions (ii) and (iii). 
~\\~\\
\textbf{Projected designs across a fixed bipartition of qubits:~} Note that even though subset phase states do not satisfy deep thermalization, it is possible to ``patch" two subset phase states (or two subset states) together so that the resulting state forms a projected design across a specific bipartition of qubits. While this property is significantly weaker than full deep thermalization, which requires such behavior across most bipartitions, it serves as a useful intermediate step toward our final construction. Furthermore, this illustrates the limitations of using existing pseudoentangled state constructions in simulating deep thermalization.

Indeed, we can ``glue" two subset phase states of $n/2$ qubits, assigning the first $n/2$ qubits to subsystem $A$, and the last $n/2$ qubits to subsystem $B$. Consider two sequences of disjoint subsets: $A_1, A_2, \cdots A_R$ and $B_1, B_2, \cdots, B_R$, such that each $A_i, B_i \subset \{0,1\}^{n/2}$ of size $M$. Define the \textit{Schmidt-patch subset phase states} as:
\begin{equation}
    \label{eq:schmidt-patch}
    \ket{\psi_{M,R,f}}=\dfrac{1}{M\sqrt{R}} \sum_{j=1}^{R} \sum_{x\in A_j,~y\in B_j} (-1)^{f(x,y)}\ket{x}_A\ket{y}_B,
\end{equation}
where $f:\{0,1\}^n\mapsto \{0,1\}$ is a global pseudorandom function. Omitting the phase yields the corresponding Schmidt-patch subset state. 

These states can be efficiently prepared from the all-zero state $\ket{0}^{\otimes n}$ by (1) preparing $r=\log_2 R$ Bell pairs between qubits $(1,n/2+1), (2, n/2+2), \cdots (r, n/2+r)$, thus entangling subsystems $A$ and $B$ with $r$ bits of entanglement across the A-B cut, and (2) applying subset phase circuits on both halves (see Sec.~\ref{app:schmidtpatchprep} of the SM \cite{supmat}).  For $r \sim \log^2 n$, the Schmidt patch construction ensures that the global state has very low entanglement across the fixed cut, while all other cuts inherit low entanglement from the underlying subset phase states. 

By construction, a projective measurement on one half (say, subsystem $B$) in the computational basis collapses the global state on to a subset phase state on $A$. Since subset phase states are known to be pseudoentangled for $M=2^{\Omega(\log^2 n)}$, this projected ensemble is trivially pseudoentangled \footnote{We use standard complexity-theoretic notation: $f(n) = O(g(n))$ if $f(n) \leq c \cdot g(n)$ for all $n \geq n_0$, and $f(n) = \Omega(g(n))$ if $f(n) \geq c \cdot g(n)$ for all $n \geq n_0$, for constants $c, n_0 > 0$. A function $f(n)$ is negligible, denoted $f(n) = \negl(n)$, if it decays with $n$ ( e.g.\ $f(n) = 1/\mathsf{poly}(n)$) for all sufficiently large $n$.}. The non-trivial claim is that the global state $\ket{\psi_{M,R,f}}$ is computationally indistinguishable from a Haar-random pure state on $n$-qubits. This is formalized as follows:
\begin{theorem}
\label{fact:schmidt-patch-haar}
    Suppose $M^{3/2}R/N < 1$, where $N=2^n$. Then for any $t=\poly(n)$, 
    \begin{align*}
    \left\|\mathbb{E}_{A, B, f} \ket{\psi_{M,R,f}}\bra{\psi_{M,R,f}}^{\otimes t} - \E_{\phi\sim\mathrm{Haar}(\C^{N})}\ket{\phi}\bra{\phi}^{\otimes t}\right\|_1\\
    \leq O\left(\dfrac{t^3}{M}\right)+O\left(\dfrac{t^2}{\sqrt{N}}\right).
    \end{align*}
\end{theorem}
The proof is given in Sec.~\ref{app:SchmidtPatchHaarProof} of the SM \cite{supmat}. An analogous result also holds for subset states. Thus, even when $\log_2 M$ and $r$ are low ($\sim \log^2 n$),  the global Schmidt-patch ensemble is pseudorandom in a strong sense: its $t$-fold moments match Haar up to negligible error while its entanglement is minimal across all cuts. The proof follows from careful combinatorial analysis of the coverage properties of the special family of sets that correspond to Schmidt-patched states.
We also extend this construction to variants with unequal patch sizes and other, possibly asymmetric (but still specific) bipartitions.

To spoof deep thermalization across most cuts, a natural place to gain intuition from is the glued unitary construction of \cite{schuster2025randomunitariesextremelylow}, which shows how to prepare PRUs (unitaries that take any state to a pseudorandom state) in low depth with low entanglement. However, analyzing the projected design properties of states prepared by such circuits is technically challenging. The complication arises because projected ensembles are dependent on the choice of measurement basis. A more direct approach is to design a minimal family of states which are tailored to computational basis measurements. The idea is to modify the standard binary phase state such that the Boolean function is structured according to a two-layer large-brick brickwork circuit, taking inspiration from the glued brickwork model introduced in \cite{schuster2025randomunitariesextremelylow}.
~\\~\\
\textbf{Large-brick phase states:~}We consider a “glued brickwork” quantum circuit made from two staggered layers of local phase gates, where each gate (or ``brick”), acts on a small block of $2b$ adjacent qubits, as shown in Figure \ref{fig:largebricks1}. Within each layer, the bricks are non-overlapping and tile the qubit chain uniformly. The second layer is shifted by exactly $b$ qubits relative to the first, so that every brick overlaps with bricks from the other layer on exactly half of its qubits. The circuit wraps around the ends of the chain, forming a ring-like structure with periodic boundary conditions. 
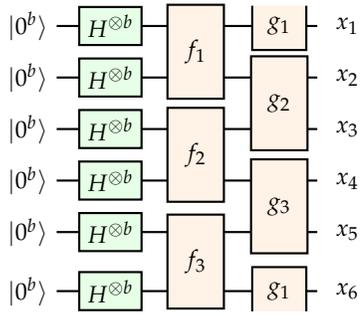
\begin{figure}[h!]
    \centering
    \begin{tikzpicture}[thick]
\tikzset{
operator/.style = {draw,fill=white!90!green,minimum size=1.5em},
operator2/.style = {draw,thick,fill=white!90!orange,minimum height=1.25cm, minimum width=0.75cm}
}
\matrix[row sep=0.1cm, column sep=0.4cm] (circuit) { 
\node (q1) {$|0^b\rangle$};  
& [-0.5cm] 
& \node[operator] (H11) {$H^{\otimes b}$}; 
& [0.2cm]\node[] (U11) {};                                     
& [0.5cm]\node[] (U12) {};   
& 
\coordinate (end1); \\
\node (q2) {$|0^b\rangle$};    
&                                     
&\node[operator] (H21) {$H^{\otimes b}$}; 
&\node[](U21){};                                      
&\node[] (U22) {};   
&\coordinate (end2);\\
\node (q3) {$|0^b\rangle$};    
&                                      
&\node[operator] (H31) {$H^{\otimes b}$}; 
& \node[](U31){};                                      
& \node[] (U32) {};   
&\coordinate (end3);\\

\node (q4) {$|0^b\rangle$};    
&                                      
&\node[operator] (H41) {$H^{\otimes b}$}; 
& \node[](U41){};                                      
& \node[] (U42) {};   
&\coordinate (end4);\\

\node (q5) {$|0^b\rangle$};    
&                                      
&\node[operator] (H51) {$H^{\otimes b}$}; 
&\node[](U51){};                                      
&\node[] (U52) {};   
&\coordinate (end5);\\

$\vdots$ & $\vdots$ & $\vdots$ & $\vdots$ & $\vdots$;\\

\node (q6) {$|0^b\rangle$};    
&                                      
&\node[operator] (H61) {$H^{\otimes b}$}; 
&\node[](U61){};                                      
&\node[] (U62) {};   
&\coordinate (end6);\\
};

\begin{pgfonlayer}{background}
\draw[thick] (q1) -- (end1)  
(q2) -- (end2) 
(q3) -- (end3)
(q4) -- (end4)
(q5) -- (end5)
(q6) -- (end6);

\node[operator2] at ($(U11)!0.5!(U21)$) {$f_1$};
\node[operator2] at ($(U31)!0.5!(U41)$) {$f_2$};
\node[operator2] at ($(U51)!0.5!(U61)$) {$f_3$};

\node[operator2] at ($(U22)!0.5!(U32)$) {$g_2$};
\node[operator2] at ($(U42)!0.5!(U52)$) {$g_3$};

\draw[draw=black, thick, fill=white!90!orange]  ($(U62.south west)+(-2.4mm,-1.5mm)$) 
-- ($(U62.north west)+(-2.4mm,2mm)$)
-- ($(U62.north east)+(2.4mm,2mm)$)
-- ($(U62.south east)+(2.4mm,-1.5mm)$);

\draw[draw=black, thick, fill=white!90!orange]  ($(U12.north west)+(-2.4mm,1.5mm)$)
-- ($(U12.south west)+(-2.4mm,-2mm)$) 
-- ($(U12.south east)+(2.4mm,-2mm)$)
-- ($(U12.north east)+(2.4mm,1.5mm)$);

\node[] at ($(U12)+(0mm,0mm)$) {$g_1$};
\node[] at ($(U62)+(0mm,0mm)$) {$g_1$};

\node[] at ($(end1)+(4mm,0mm)$) {$x_1$};
\node[] at ($(end2)+(4mm,0mm)$) {$x_2$};
\node[] at ($(end3)+(4mm,0mm)$) {$x_3$};
\node[] at ($(end4)+(4mm,0mm)$) {$x_4$};
\node[] at ($(end5)+(4mm,0mm)$) {$x_5$};
\node[] at ($(end6)+(4mm,0mm)$) {$x_{6}$};
\end{pgfonlayer}
\end{tikzpicture}
    \caption{\smaller Two-layered brickwork circuit for preparing a large-brick phase state. The circuit applies $m$ bricks to the equal superposition state of $n$ qubits. Each brick is of size $2b$ and applies a phase gate in the computational basis, constructed out of pseudorandom functions $f_i, g_i \in \{0,1\}^{2b}\mapsto\{0,1\}$, such that each $x_i\in \{0,1\}^b$. In this figure, $m=3$ and the final state is of $n=2mb=6b$ qubits.}
    \label{fig:largebricks1}
\end{figure}
Each brick applies a local phase gate to its input, determined by a pseudorandom function sampled independently for each brick: either $f_i:\{0,1\}^{2b}\mapsto \{0,1\}$ for bricks in the first layer or $g_i:\{0,1\}^{2b}\mapsto \{0,1\}$, for the second. These act as $U_{f_j} |x\rangle = (-1)^{f_j(x)} \ket{x}$, on any $2b$-bit string $x$. Here $j\in [1,m]$, as each layer contains exactly $m=n/(2b)$ bricks \footnote{Without loss of generality, we assume $n$ is even and $b$ is a divisor of $n/2$}. Starting from the uniform superposition state of $n$-qubits, this circuit prepares a binary phase state on each basis string, built from a structured sum of local, pseudorandom contributions. For $b=\Omega(\log^2 n)$, we call these states \emph{large-brick phase states}, formally define as
\begin{equation}\label{eq:largebrickPRS}
        \ket{\psi_{f_{1:m},g_{1:m}}} = \dfrac{1}{2^{n/2}}\sum_{x_1,\dots,x_{2m}\in\B^b}(-1)^{\Phi(x)}\ket{x_1.x_2.\;\cdots\;.x_{2m}},
\end{equation}
where the phase function,
$$\Phi(x)=\sum _{j\in[1,m]}f_j(x_{2j-1}.x_{2j}) + \sum_{j\in[1,m]}g_j(x_{2j-2}.x_{2j-1}).
$$
Here, each $x_i\in \{0,1\}^b$ with $x_0:=x_{2m}$ (periodic boundary conditions), and $x.y$ denotes the concatenation of the two bit strings $x$ and $y$. Each brick acting on $2b\sim \log^2 n$ qubits can be efficiently compiled into a quantum circuit of depth approximately $\sim \log \log n$, using gates of constant locality. While these gates may be geometrically non-local, the overall circuit remains sparse and efficient to implement. The pseudorandom functions used to define the bricks are constructed using standard techniques from cryptographic hardness assumptions—such as the presumed difficulty of lattice problems—and can be instantiated from any quantum-secure one-way function \cite{peikert2008trapdoor}. In particular, the existence of quantum-secure one-way functions implies the existence of pseudorandom functions that we use in the circuit in Fig.~\ref{fig:largebricks1}. We now prove that large-brick phase states satisfy Conditions (i)-(iii). 

We first show that these states are computationally indistinguishable from Haar-random states and are pseudoentangled. Formally, we prove \Cref{fact:LargeBrickGlobalPseudorandom}.
\begin{theorem}[Large-brick phase states are pseudorandom and pseudoentangled]\label{fact:LargeBrickGlobalPseudorandom}
    Suppose, for large-brick phase states defined in Eq.~\eqref{eq:largebrickPRS}, the functions $f_1,\dots,f_m$ and $g_1, \dots, g_m$ are chosen to be genuinely random. Then the ensemble moments match the moments of the Haar measure up to a negligible difference:
    \begin{align*}
\left\|\E_{f_{1:m},g_{1:m}}\dyad{\psi_{f_{1:m},g_{1:m}}}^{\otimes t} - \E_{\phi\sim\mathrm{Haar}(\C^N)}\dyad{\phi}^{\otimes t}\right\|_1 \\
\leq O\left(\frac{t^2n^2}{2^b}\right).
    \end{align*}
    Moreover, the states are pseudoentangled, having entropy $O(b)$ along 1D geometric cuts.
\end{theorem}

While the detailed proof of Theorem~\ref{fact:LargeBrickGlobalPseudorandom} can be found in Sec.~\ref{subsec-app:lbps-close-to-haar} of the SM \cite{supmat}, we provide a sketch here. Note that for these states, the global phase is generated by local functions acting on bricks of $2b$ bits. For genuinely random functions $f, g$, we analyze the expected $t$-moment operator and show that it matches the Haar moment on the \emph{birthday subspace}---the set of $t$-tuples where no $b$-bit substring repeats (collision-free subspace). This relies on the key observation that random phases contributed by each brick average to zero unless the local inputs match pairwise, which only occurs when the $t$-tuples are globally related by a permutation. Thus, the moment operator reproduces Haar structure on most of its weight, yielding trace distance closeness with error $O(t^2 n^2 / 2^b)$. This is $\negl(n)$ for large bricks, i.e.\ $b\sim \log^2 n$. 

To achieve computational pseudorandomness, we replace the truly random local functions with quantum-secure pseudorandom functions (PRFs). Since PRFs are indistinguishable from random functions to any efficient quantum algorithm, large-brick phase states are computationally indistinguishable from Haar random states. The existence of quantum-secure PRFs can be based on a variety of conjectures, for instance, Learning With Errors (LWE) or Learning Parity with Noise (LPN) \cite{PeikertWaters2008}.

Furthermore, these states satisfy a (quasi) area-law for entanglement. Specifically, the entanglement entropy across any cut is upper bounded by $O(|B| \cdot b)$, where $|B|$ is the number of bricks intersected by the cut (i.e., the ``area'' of the boundary), and $b$ is the brick  size. This bound follows from a simple geometric observation: only bricks that straddle the cut can contribute to the entanglement, and each such brick adds at most $O(b)$ entropy. So, the entanglement entropy across any 1D geometric cut is $O(b)$, which is $\sim \log^2 n$, for large-brick phase states. 
~\\ ~\\ 
\textbf{Projected designs across almost all cuts:~}Large-brick phase states yield projected designs across all but a negligible fraction of the cuts of the state. The formal proof is given in Sec.~\ref{subsec-app:lbps-projected-design} of the SM \cite{supmat} while we outline the central ideas here.

We show that for a typical bipartition of the $n$ qubits, each $b$-qubit brick is split into two segments of lengths $k_j$ and $b - k_j$, where $k_j$ is sharply concentrated around $b/2$ by a standard binomial tail bound. A union bound shows that, with high probability, \emph{all} bricks are (roughly) evenly split up, i.e.\ each segment has length $[b/4, 3b/4]$. We then consider a computational basis measurement on one side of the cut. The post-measurement state on the remaining qubits has a phase determined by the measurement outcome, resulting in a new ensemble of states. 

Crucially, the post-measurement phase retains the brickwork structure, and the resulting state is still a large-brick phase state: the effective phase functions are obtained by fixing part of the input bits (those on the measured side) in the original pseudorandom functions. As long as the remaining brick sizes are $\geq \Omega(\log^2 n)$, which holds with high probability (from the binomial concentration bounds), these functions remain pseudorandom. So, following \Cref{fact:LargeBrickGlobalPseudorandom}, the projected ensemble continues to satisfy pseudorandomness (i.e., Haar indistinguishability) and low entanglement (i.e., pseudoentanglement) properties. Interestingly, this also includes all but a $1/\log^2 n$ fraction of the 1D geometric cuts. Since the entanglement entropy is low across precisely these cuts, we conclude that large-brick phase states continue to exhibit the projected design property even across most of the ``pesudoentangled" cuts.
~\\ ~\\ 
\textbf{Towards pseudorandomness beyond BQP:~}We note that the ensemble constructed in this work may potentially remain pseudorandom and pseudoentangled against a qualitatively different adversary than what is usually considered. Standard computational pseudorandomness is defined against distinguishers equipped with polynomial-time quantum computers, i.e., those in \textsf{BQP}. It is also known that any ensemble becomes distinguishable from Haar when \textsf{BQP} is augmented with a \textsf{PP} oracle ~\cite{kretschmer}, which can postselect on exponentially rare events. This raises a natural question: can pseudorandomness and pseudoentanglement persist in the large complexity landscape between \textsf{BQP} and $\textsf{BQP}^{\textsf{PP}}$?

Our work takes a first step in this direction by considering \textsf{BQP} distinguishers equipped with the power of non-collapsing measurements. The key insight is that projected designs appear to be robust against such adversaries. However, prior constructions of pseudoentangled states, which are not projected designs, such as subset~\cite{giurgicatiron2023pseudorandomnesssubsetstates} and subset-phase states~\cite{aaronson_et_al:LIPIcs.ITCS.2024.2}, are distinguishable by adversaries capable of performing $\textsf{poly}(n)$ non-collapsing measurements. 

In contrast, pseudorandom phase states~\cite{JiLiuSong2018, ananth2023pseudorandomfunctionlikequantumstate, aaronson_et_al:LIPIcs.ITCS.2024.2} can be shown to form projected designs across every cut, and remain indistinguishable from Haar even under such measurements---though they are not pseudoentangled due to high entanglement. Our large-brick phase states (Fig.\ref{fig:largebricks1}) achieve both: they maintain low entanglement while forming projected designs across almost all bipartitions. Thus, an adversary that selects a random cut and performs non-collapsing measurements will, with high probability, observe pseudoentangled residual states (see Sec.~\ref{sec:pseudorandomness-beyond-bqp} of the SM~\cite{supmat}). However, any such robustness claim relies on the cryptographic assumption that pseudorandom functions remain secure against such adversaries, and we leave open the plausibility of this conjecture for
future investigation.

Finally, our model of the adversary, making non-collapsing measurements, is a special case of the complexity class \textsf{PDQP} (Product Dynamical Quantum Polynomial time) \cite{aaronson2016spacejustabovebqp,miloschewsky2025newlowerboundsquantumcomputation},  defined as the class “just above \textsf{BQP}”. Our results suggest that pseudoentanglement may persist even against \textsf{PDQP}-like adversaries, providing a new physical motivation—from the perspective of deep thermalization—for exploring pseudorandomness beyond \textsf{BQP}.
~\\~\\
\textbf{Discussion:~}Experimental and numerical studies have shown that many-body quantum systems thermalize rapidly and with surprisingly low resource overhead \cite{choi2023preparingrandomstates,cotler2023emergent}. Since computation is ultimately a physical process, and nature itself can be viewed as a quantum computer, these observations call for theoretical explanations from the perspective of quantum computation. Our findings provide such a foundation, showing that deep thermalization can indeed be simulated by quantum computers using minimal resources. Unlike prior theoretical results, which require highly idealized conditions (e.g., infinite bath size or constant-copy distinguishers), our construction accommodates finite baths and allows an efficient observer to request polynomially many copies of the post-measurement state.  

The circuits required to prepare our states have depth $\log^2 n$, which exceeds what polynomial-time classical algorithms can efficiently simulate. This implies that classically reproducing computational deep thermalization would require circuits of at least quasi-polynomial depth ($\sim 2^{\log^2 n}$).

From the perspective of computational pseudorandomness, our results provide one of the final pieces in the broader program to derandomize the Haar measure under computational assumptions \cite{JiLiuSong2018, aaronson_et_al:LIPIcs.ITCS.2024.2, schuster2025randomunitariesextremelylow, ma2025constructrandomunitaries}. They also highlight how computational indistinguishability can reconcile low-complexity states with emergent thermal behavior, offering new insights at the intersection of complexity theory and physics. 

A natural next step is to generalize these ideas to deep thermalization at finite temperatures. In this regime, the projected ensemble becomes the so-called Scrooge ensemble—a ``thermally shifted'' Haar distribution~\cite{PhysRevX.14.041051}. Even reproducing its second moments using shallow quantum circuits remains an open problem.
~\\~\\
\textit{Note added:~} Concurrently and independently, Ref.~\cite{cui2025unitarydesignsnearlyoptimal} uses a construction similar to large brick phase states to prepare fast $t$-designs and pseudorandom states. As noted in \cite{cui2025unitarydesignsnearlyoptimal},  if the pseudorandom functions are replaced by $t$-wise independent functions, then the construction in \Cref{fact:LargeBrickGlobalPseudorandom} gives $t$-designs that are optimal in certain parameters. However, Ref.~\cite{cui2025unitarydesignsnearlyoptimal} does not discuss implications to deep thermalization.
~\\~\\
We thank Bill Fefferman, Adam Bouland, Thomas Schuster, Robert Huang, Manuel Endres, Matteo Ippoliti, Tobias Haug, Supartha Podder, and David Miloschewskya for helpful comments. Shantanav Chakraborty acknowledges funding from the Ministry of Electronics and Information Technology (MeitY), Government of India, under Grant No. 4(3)/2024-ITEA. Shantanav Chakraborty also acknowledges support from Fujitsu Ltd, Japan, and IIIT Hyderabad. Soonwon Choi acknowledges funding from the NSF CAREER award 2237244. 

\bibliographystyle{apsrev4-2}
\bibliography{refs.bib,refs_SC}

\cleardoublepage

\cleardoublepage
\onecolumngrid
\setcounter{equation}{0}
\setcounter{footnote}{0}
\setcounter{figure}{0}
\setcounter{section}{0}
\thispagestyle{empty}
\renewcommand{\theequation}{S\arabic{equation}}
\renewcommand{\thetable}{S\arabic{table}}
\renewcommand{\thefigure}{S\arabic{figure}}
\renewcommand{\thetheorem}{S\arabic{theorem}}
\renewcommand{\thelemma}{S\arabic{lemma}}
\renewcommand{\thecorollary}{S\arabic{corollary}}
\renewcommand{\thesection}{S-\Roman{section}}
\theoremstyle{plain}           
\newtheorem{claim}{Claim}
\begin{center}
\textbf{\large Supplemental Material for ''Fast computational deep thermalization''}\\
\vspace{2ex}
Shantanav Chakraborty, Soonwon Choi, Soumik Ghosh, and Tudor Giurgic\u{a}-Tiron
\vspace{2ex}
\end{center}

\renewcommand\thesection{S\arabic{section}}
\renewcommand\thesection{S\arabic{section}}
\renewcommand\thefigure{S\arabic{figure}}

Here, we provide detailed derivations of the results of the main article.

\section{Notations and preliminaries}
\noindent Let us first establish some basic notation:

\begin{definition}
    Given a set $X$, denote by $\powerset{X}$ (alternatively, denoted $2^X$) the power set of $X$, i.e. the set of all subsets of $X$. For our practical uses, let us denote by $\calP_{\geq\polylog}(X)$ the set of subsets of $X$ of size $2^{\Omega(\polylog\log\abs{X})}$.
\end{definition}

\begin{definition}
    Given a subset $S \in \powerset{[N]}$, we define the $(S, t)$-symmetric subspace $\Sym{t}{S}$ as the symmetric subspace of $t$ copies of the space spanned by the vectors $\{\ket{j}\}_{j\in S}$.
\end{definition} 
~\\
\textbf{Higher moments of Haar random states:~}It is well known (see for instance \cite{harrow2013church}) that the $t$-th moment of a Haar random state on a $d$-dimensional Hibert space $\mathcal{H}$ can be written as a projection $\Pi_{\mathrm{Sym}^t(\mathcal{H})}$ onto the symmetric subspace of $\mathcal{H}^{\otimes t}$. That is,
\begin{equation}
\rho^{(t)}_{\mathrm{Haar}}=\mathbf{E}_{\phi\sim\mathrm{Haar}(\mathcal{H})}\dyad{\phi}^{\otimes t}=\int_{\phi\sim\mathrm{Haar}(\mathcal{H})}d\phi~\dyad{\phi}^{\otimes t}=\dfrac{\Pi_{\mathrm{Sym^t}(\mathcal{H})}}{\mathrm{dim}~\mathrm{Sym}^t(\mathcal{H})},
\end{equation}
where,
\[
\mathrm{dim}~\mathrm{Sym}^t(\mathcal{H})=\binom{d+t-1}{t}.
\]
~\\~\\
\textbf{Pseudoentangled states:~}An ensemble $\mathcal{E}$ of $n$-qubit pure states is said to be pseudoentangled \cite{aaronson_et_al:LIPIcs.ITCS.2024.2, bouland_et_al:LIPIcs.CCC.2024.21}, if every state $\ket{\psi}\in\mathcal{E}$ is (i) efficiently preparable, (ii) is computationally pseudorandom, and has (iii) low entanglement. Let us expand on each of these points. Point (i) implies that $\ket{\psi}$ can be prepared by a quantum circuit of depth at most $\poly(n)$. For any $t=\poly(n)$, the ensemble $\mathcal{E}$ is computationally pseudorandom if any polynomial-time quantum algorithm $\mathcal{C}$ with access to $t$ copies of $\ket{\psi}\in\mathcal{E}$ cannot distinguish it from $t$-moments of the Haar ensemble. More concretely,  
\begin{equation}
    \left|\mathcal{C}\left(\mathbb{E}_{\psi\sim \mathcal{E}}\dyad{\psi}^{\otimes t}\right) -\mathcal{C}\left(\mathbb{E}_{\phi\sim\mathrm{Haar}}\dyad{\phi}^{\otimes t}\right)\right|\leq \negl(n),
\end{equation}
where $\negl(n)$ is any function that decays with $n$ (e.g.~$1/\poly(n)$).
For point (iii), the state $\ket{\psi}$ should have low entanglement across spatially local (or geometrically local) bipartitions of qubits, despite being computationally pseudorandom. Typically, the entanglement across these cuts scales as some function of $n$ (say $\Theta(f(n))$), which is required to be low in comparison to another state having high entanglement (say $\Theta(g(n))$). The difference between the two is what is referred to as the \textit{entanglement gap}. In the main article, we compare the entanglement content of the states we construct with Haar random states. For instance, large-brick phase states are pseudoentangled states with entanglement across geometrically local cuts being as low as $\Omega(\log^2 n)$. The entanglement is (sub)-exponentially lower than Haar random states.

\section{Subset phase states are not projected designs}
\label{section: subset states are not projected designs}
Consider the states considered in Ref.~\cite{aaronson_et_al:LIPIcs.ITCS.2024.2}, which demonstrated an example of a PRS ensemble which also presented undetectably low entanglement across all cuts with high probability. The construction relies on pseudorandom subset-phase states:
\begin{definition}[Subset-phase state \cite{aaronson_et_al:LIPIcs.ITCS.2024.2}]
    \begin{equation}
    \label{eq:PRS}
    \ket{\psi_{M,R,f}} \equiv \frac{1}{\sqrt{M}}\sum_{x\in S}(-1)^{f(x)}\ket{x}\quad \in \mathbb C^{2^n}\,,
\end{equation}
where $S \in \binom{[N]}{M}$ is a pseudorandom subset of the computational basis of size $M= \abs{S}=2^{\polylog(n)}$, and $f$ is a classical one-way function.
\end{definition}

We will also work in the hybrid picture, in which pseudorandom quantities such as $S$ and $f$ are treated as genuinely uniformly random, aiming to show that this leads to states with strong statistical properties.

\begin{observation}
By design, a state like in Eq.~\eqref{eq:PRS} cannot build projected designs because it is constructed to have low entanglement across (any) cuts. Birthday statistics tell us that collapsing a subsystem $B$ of size $n-O(n)$ will leave us with a computational basis state in $A$, i.e. a depth-one state, with probability  $1-O(2^{\polylog(n)-\abs{B}}) = 1-\textsf{negl}(n)$. It is well-known that such states cannot form designs, since there is a $\Omega(nt)$ lower bound on the depth of the elements of a $t$-design. The same argument holds for subset states \cite{giurgicatiron2023pseudorandomnesssubsetstates}.
\end{observation}

\section{Generic subset distributions for Haar-like state ensembles}

In this section, we aim to relax the uniform subset distribution in the standard subset and subset-phase state constructions. Our goal is to find technical conditions that a distribution $\calD$ over (large enough) subsets of $[N]$ should satisfy such that the corresponding subset-phase and subset state ensembles remain close to Haar. 

First, we study the conditions for subset-phase states, assuming the phase function is globally uniformly random, but allowing a generic distribution over subsets:

\begin{lemma}[Subset distributions for Haar-like subset-phase states]\label{lemma:subsetstatistics}
    For $N=2^n$, consider a distribution $\mathcal D$ over the polylog-sized power set $\calP_{\geq\polylog}([N])$. The information-theoretic condition for phase-subset states induced by $\mathcal D$ to form $t$-designs, for $t=\poly(n)$, is:
    \begin{equation}
        \label{eq:tdspprs}
        \norm{\E_{\phi\sim\text{Haar}(N)}\dyad{\phi}^{\otimes t} - \E_{S\sim\mathcal{D}}\E_{f\sim \B^{[N]}} \dyad{\psi_{f, S}}^{\otimes t}}_1 = O(g(n)) = \negl(n)\,.
    \end{equation}
    Then, an equivalent condition in terms of distributions over subsets of size $t$ is:
    \begin{equation}
        \label{eq:inclusiontest}
       \sum_{T\in\binom{[N]}{t}} \abs{\mathbb{P}_\mathcal{D}(T) - \mathbb{P}_0(T)} = O(g(n))=\negl(n)\,.
    \end{equation}
    The function $g(n)$ in \eqref{eq:tdspprs} and \eqref{eq:inclusiontest} is the same. This reduces the trace distance to the classical total variation distance between the uniform distribution $\mathbb P_0(T)=\binom{N}{t}^{-1}$, and the distribution induced by $\mathcal D$ (i.e. sample a subset $S$ from $\mathcal{D}$, then sample a random subset of size $t$ uniformly):
    \begin{equation}
        \label{eq:inclusionprobability}
        \mathbb P_\mathcal{D}(T) = \mathbb E_{S\sim \mathcal D}\frac{\mathbf{1}[T\subset S]}{\binom{\abs{S}}{t}}\,.
    \end{equation}
    In other words, from the `local' perspective of most small subsets of size $t$, it looks as if they are included uniformly in the overall distribution. In the above, $\mathbf{1}[X]$ is the indicator variable taking value 1 if event $X$ occurs, and 0 otherwise.\\
\end{lemma}

\begin{proof}
We note that Proposition 1 in \cite{feng2025dynamics} proves a similar result. For a subset distribution $\calD$, define the following three matrices which will end up being close:
    \begin{align}
        \rho_0 &\equiv \E_{\phi\sim \Haar{[N]}} \dyad{\phi}^{\otimes t} = \frac{\Pi_{\Sym{t}{[N]}}}{\dim~\Sym{t}{[N]}}\\
        \rho_\calD^{\text{Haar}} &\equiv \E_{S\sim \calD}\E_{\phi_S\sim \Haar{S}} \dyad{\phi_S}^{\otimes t} = \E_{S\sim \calD}\frac{\Pi_{\Sym{t}{S}}}{\dim\Sym{t}{S}}\\
        \rho_\calD^{\text{phase}} &\equiv \E_{S\sim \calD}\E_{f} \dyad{\psi_{f,S}}^{\otimes t}
    \end{align}
    First, let us show that $\rho_\calD^{\text{Haar}}$ and $\rho_\calD^{\text{phase}}$ are close. This follows from a triangle inequality on term-by-term differences, when we use the fact that:
    \begin{equation}
        \norm{\E_f \dyad{\psi_{f,S}}^{\otimes t} - \frac{1}{\dim\Sym{t}{S}}\Pi_{\Sym{t}{S}}}_1 = O\left(\frac{t^2}{\abs{S}}\right)\,.
    \end{equation}
    This is simply the main result of \cite{brakerski2019pseudo} applied to the set $S$ as opposed to $[N]$. Then the triangle inequality gives us that:
    \begin{align}
        \norm{\rho_\calD^{\text{Haar}} - \rho_\calD^{\text{phase}}}_1 &\leq \E_{S\sim \calD} \norm{\E_f \dyad{\psi_{f,S}}^{\otimes t} - \frac{1}{\dim\Sym{t}{S}}\Pi_{\Sym{t}{S}}}_1&\\
        &\leq \E_{S\sim \calD} O\left(t^2/\abs{S}\right)&\\
        &\leq O\left(\frac{t^2}{\min_{S \in \mathop{\text{supp}}(\calD)} \abs{S}}\right)&\\
        &= \negl(n)&\text{ since }\mathop{\text{supp}}(\calD) \subset \powerset{[N]}_\text{eff}\,.
    \end{align}
    Now, what remains is to find the conditions under which $\rho_\calD^{\text{Haar}}$ is close to the full Haar average $\rho_0$. First, notice that $\rho_\calD^{\text{Haar}}$ (and indeed all of the single-subset terms) commute with $\rho_0$ since they share the type basis of the full symmetric subspace. In other words, let the orthonormal basis of $\Sym{t}{S}$ be made of the single-type vectors $\ket{\theta}$ for $\theta \in \mathop{\text{Types}\/}\left(S, t\right)$ --- where a type $\theta$ represents the frequencies of each $s\in S$. Specifically, if $S=\{s_1, \dots, s_m\}$, then $\theta=(k_{s_1}, k_{s_2}, \dots, k_{s_m})$ with $k_{s_1}+\dots+k_{s_m}=t$ nonnegative integers. For such a $\theta$ type, the corresponding basis vector $\ket{\theta}$ is proportional to all type permutations, namely:
    \begin{equation}
    \ket{\theta}=\ket{\left(k_{s_1}, \dots, k_{s_m}\right)} \propto \sum_{\sigma \in \mathfrak{S}_t}\sigma\ket{\underbrace{s_1,\dots,s_1}_{k_{s_1}\text{ times}},\;\underbrace{s_2,\dots,s_2}_{k_{s_2}\text{ times}},\dots\dots,\underbrace{s_m,\dots,s_m}_{k_{s_m}\text{ times}}}\,.
    \end{equation}
    With this choice of orthonormal basis, we have that $\Pi_{\Sym{t}{S}} = \sum_{\theta\in\mathop{\text{Types}\/}\left(S, t\right)}\dyad{\theta}$. Our goal is to have near-equality between $\rho_\calD^{\text{Haar}}$ and $\rho_0$, which means:
    \begin{align}
        \rho_0 &\approx_\text{1-norm} \rho_\calD^{\text{Haar}} \\
        \frac{1}{\binom{N+t-1}{t}}\Pi_{\Sym{t}{[N]}} &\approx_\text{1-norm} \E_{S\sim\calD} \frac{1}{\binom{\abs{S}+t-1}{t}}\Pi_{\Sym{t}{S}} \\
        \frac{1}{\binom{N+t-1}{t}}\sum_{\theta\in\mathop{\text{Types}\/}\left([N], t\right)}\dyad{\theta} &\approx_\text{1-norm} \E_{S\sim\calD} \frac{1}{\binom{\abs{S}+t-1}{t}}\sum_{\varphi\in\mathop{\text{Types}\/}\left(S, t\right)}\dyad{\varphi}\,.
    \end{align}
    Here $A\approx_{\text{1-norm}}B$ implies that operators $A$ and $B$ are within $\negl(n)$ in 1-norm distance. 
    Notice that both sides here are diagonal in the same basis of $[N]$-types, since $\mathop{\text{Types}\/}\left(S, t\right) \subset \mathop{\text{Types}\/}\left([N], t\right)$ for any $S\subset [N]$, so it is enough to evaluate the diagonal elements:
    \begin{align}
        \mel{\theta}{\rho_\calD^{\text{Haar}}}{\theta} &= \E_{S\sim\calD}\frac{\mathop{\mathbf{1}}\left[\theta \in \mathop{\text{Types}\/}\left(S, t\right)\right]}{\binom{\abs{S}+t-1}{t}} & \text{for } \theta\in \mathop{\text{Types}\/}\left([N], t\right)\\
        &= \E_{S\sim\calD}\frac{\mathop{\mathbf{1}}\left[\mathop{\text{supp}\/}(\theta) \subset S\right]}{\binom{\abs{S}+t-1}{t}}\,, &
    \end{align}
    where we naturally denote by $\mathop{\text{supp}\/}(\theta)$ the set of computational basis elements which have nonzero frequency in the type $\theta$. Therefore, the trace distance between the two matrices is:
    \begin{align}
        \frac12\norm{\rho_0 - \rho_\calD^{\text{Haar}}}_1 &= \frac12\sum_{\theta\in\mathop{\text{Types}\/}\left([N], t\right)} \abs{\frac{1}{\binom{N + t - 1}{t}} - \mel{\theta}{\rho_\calD^{\text{Haar}}}{\theta}} \\
        &= \frac12\sum_{\theta\in\mathop{\text{Types}\/}\left([N], t\right)} \abs{\frac{1}{\binom{N + t - 1}{t}} - \E_{S\sim\calD}\frac{\mathop{\mathbf{1}}\left[\mathop{\text{supp}\/}(\theta) \subset S\right]}{\binom{\abs{S}+t-1}{t}}}\,.\label{eq:fullTVdistance}
    \end{align}
    We know that because of birthday statistics, we have that the types $\mathop{\text{Types}\/}\left(S, t\right)$ are dominated by the unique types $\binom{S}{t}$, since the total number of types is $\binom{\abs{S}+t-1}{t} = \binom{\abs{S}}{t}\left(1 + O\left(\frac{t^2}{\abs{S}}\right)\right)$. The correction is negligible because, in our case, $\abs{S}=2^{\polylog(n)}$ while $t=\poly(n)$. Let us show that restricting the above distance to the set of unique types $\binom{[N]}{t}$ is enough to have overall small trace distance. Assume then that we have:
    \begin{equation}
        \label{eq:condition1}
        \sum_{\theta\in\binom{[N]}{t}} \abs{\frac{1}{\binom{N}{t}} - \E_{S\sim\calD}\frac{\mathop{\mathbf{1}}\left[\theta \subset S\right]}{\binom{\abs{S}}{t}}} = \negl(n)\,,
    \end{equation}
    and let us see why this implies that the full sum \eqref{eq:fullTVdistance} is also negligible. This is straightforward: the probability mass in $\rho_0$ outside of the unique types adds up to:
    \begin{equation}
        \label{eq:addition1}
        \sum_{\theta \in \mathop{\text{Types}\/}\left([N], t\right) \setminus \binom{[N]}{t}} \mel{\theta}{\rho_0}{\theta} = 1-\frac{\binom{N}{t}}{\binom{N+t-1}{t}} = O\left(\frac{t^2}{N}\right) = \negl(n)\,.
    \end{equation}
    Similarly, the probability mass in $\rho_\calD^{\text{Haar}}$ on this space of non-unique types adds up to:
    \begin{align}
        \sum_{\theta \in \mathop{\text{Types}\/}\left([N], t\right) \setminus \binom{[N]}{t}} \mel{\theta}{\rho_\calD^{\text{Haar}}}{\theta} &= 1 - \sum_{\theta \in \binom{[N]}{t}} \mel{\theta}{\rho_\calD^{\text{Haar}}}{\theta}& \\
        &=1 - \sum_{\theta \in \binom{[N]}{t}} \mel{\theta}{\rho_0}{\theta} + \negl(n) & \text{ by triangle inequality using \eqref{eq:condition1}} \\
        &= \negl(n) & \text{ because of \eqref{eq:addition1}.}
    \end{align}
    Therefore, both matrices have negligible support outside of the unique types and the contribution to the TV distance can be at most negligible. This concludes the proof.
\end{proof}

Let us now study the case of subset states. While in the simpler case of the subset-phase states above, one can safely argue that the condition in Lemma \ref{lemma:subsetstatistics} is both necessary and sufficient, for the more combinatorially involved case of subset states we restrict ourselves to a sufficient condition:

\begin{lemma}[Subset distributions for Haar-like subset state ensembles]\label{lemma:subsetstatistics2}
    Consider a distribution $\calD$ over the large-sized subsets $\calP_{\geq\polylog}([N])$. Assume two conditions are satisfied:
    \begin{enumerate}
        \item[(a)] For any $T \in \binom{[N]}{t}$ where $t=\poly(n)$, we have that:
        \begin{equation}
            \E_{S \sim \calD} \frac{\mathbf{1}[T \subset S]}{\binom{\abs{S}}{t}} = \frac{1+O\left(g(n)\right)}{\binom{N}{t}}\,.
        \end{equation}
        \item[(b)] For any two $T,T' \in \binom{[N]}{t}$ where $t=\poly(n)$, there is a $\gamma(n)>0$ such that:
        \begin{equation}
            \E_{S \sim \calD} \frac{\mathbf{1}[T \cup T' \subset S]}{\binom{\abs{S}}{t}} \leq \frac{1+O\left(g(n)\right)}{\binom{N}{t}} \gamma(n)^{d(T,T')}\,,
        \end{equation}
        where $d(T,T')=\abs{T \setminus T'}=\abs{T\cup T'}-t=t-\abs{T\cap T'}$ is the set distance between $T$ and $T'$.
    \end{enumerate}
\end{lemma}

Then we have that the resulting subset state distribution is close to Haar in trace distance:
\begin{equation}
    \norm{\E_{S\sim \calD}\dyad{S}^{\otimes t} - \E_{\phi\sim\mathrm{Haar}([N])} \dyad{\phi}^{\otimes t}}_1 \leq O(g(n)) + O(t\gamma(n)),\quad \text{ for }t=\poly(n).
\end{equation}

\begin{proof}
    The proof follows by a simple modification of the calculation in \cite{giurgicatiron2023pseudorandomnesssubsetstates} which addresses the case of a uniform distribution over $m$-subsets. Notice that in that case the parameter $\gamma$ is $m/N=\negl(n)$.
\end{proof}

\section{Schmidt-patched states}
\label{app:schmidtpatch}
We begin by stating the general definition of a Schmidt patch subset phase state and a Schmidt patch subset state.
~\\
\begin{construction}[Schmidt patch subset phase state] Consider two Hilbert spaces $\calH_{A,B}$ of dimension $N_{A,B}$. For a rank $R$ and two sequences of disjoint $M_A$-subsets $A_1,\dots,A_R \in  \binom{[N_A]}{M_A}$ and disjoint $M_B$-subsets $B_1,\dots, B_R \in  \binom{[N_B]}{M_B}$, construct the {\bf patched phase-subset state} state on $\calH_A\otimes \calH_B$:
    \begin{align}
    \label{eq:PRScut}
    \ket{\Psi_{A_{1:R},B_{1:R},f}} &\equiv \frac{1}{\sqrt{M_AM_BR}} \sum_{j=1}^R \sum_{\substack{x \in A_j\\y \in B_j}} (-1)^{f(x, y)}\ket{x}_A\otimes \ket{y}_B\,.
\end{align}
In the main article, the state defined in Eq.~\eqref{eq:schmidt-patch} is a particular case with $M_A=M_B=M$. When taking $f$ to be the all-zero function (i.e., all phases are one), this becomes simply a Schmidt patch subset state, which Schmidt-patches together $R$ pairs of subset states:
\begin{align}
    \label{eq:patchedsubsetstate}
    \ket{\Psi_{A_{1:R}, B_{1:R}}} &\equiv \frac{1}{\sqrt{M_A M_B R}} \sum_{j=1}^R \ket{A_j}_A\otimes \ket{B_j}_B = \frac{1}{\sqrt{M_A M_B R}} \sum_{j=1}^R \sum_{\substack{x \in A_j\\y \in B_j}} \ket{x}_A\otimes \ket{y}_B\,.
\end{align}
\end{construction}

In this construction, we allow for rank $R$ across a known, fixed cut and is defined by the subsets $\{(A_j, B_j)\}_{j=1}^R$. Next, we show that these states can be prepared efficiently. 

\subsection{Preparing Schmidt-patch states}
\label{app:schmidtpatchprep}

Let $R=2^r$, $M_A=2^{k_A}$, $M_B=2^{k_B}$, $N_A=2^{n_A}$ and $N_B=2^{n_B}$, where $n_A+n_B=n$. Then, there is an efficient circuit (See Fig.~\ref{figapp:schmidt-patch-circuit}) that prepares pseudorandom samples from the uniform distribution over subset phase states from classical one-way functions. The construction is simple and uses known efficient circuits for pseudorandom permutations ($\PRP$) on each subsystem: $A$ (qubits $1$ through $n_A/2$) and $B$ (qubits $n_A/2+1$ through $n$). The steps are outlined below

    \begin{enumerate}
        \item From the $\ket{0^{n_A}}_A\ket{0^{n_B}}_B$ initial state, prepare $r$ Bell states on the qubits pairs $(1,n_A+1),(2, n_A+2),\dots (r, n_A+r)$, thus entangling subsystems $A$ and $B$ with $r$ bits of entanglement. Also apply Hadamards to the next $k_A$ qubits of subsystem $A$ and $k_B$ qubits of subsystem $B$.
        
        \item Apply independent pseudorandom $S_{N_A}$-permutations $\sigma_A=\PRP_{S_{N_A}}(k_A)$ and $S_{N_B}$ permutations $\sigma_B=\PRP_{S_{N_B}}(k_B)$ to subsystem $A$ and $B$ respectively. This results in a Schmidt patch subset state.
    \end{enumerate}

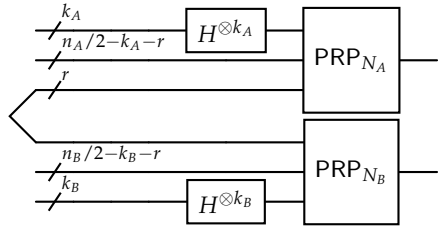
\begin{figure}[h!]
    \centering
\begin{quantikz}[row sep=0.3em]
        & \qwbundle{k_A} & & & \gate{H^{\otimes k_A}} & \gate[3]{\mathsf{PRP}_{N_A}} \\
        & \qwbundle{n_A/2-k_A-r} & & &  &  & \\
        \makeebit{} & \qwbundle{r} & & & & \\
        & &  & & & \gate[3]{\mathsf{PRP}_{N_B}} \\
        & \qwbundle{n_B/2 - k_B - r} & &  & & & \\
        & \qwbundle{k_B} & & & \gate{H^{\otimes k_B}} & \\
    \end{quantikz}
\label{figapp:schmidt-patch-circuit}
\caption{Quantum circuit for preparing a Schmidt patch subset state}
\end{figure}
In order to obtain a Schmidt patch subset phase state, apply a pseudorandom function $f = \PRF_{k_f} : [N]\mapsto \B$ via a phase oracle, as usual. 

In the main article, the state in Eq.~\eqref{eq:schmidt-patch} is a Schmidt patch subset phase state with $n_A=n_B=n/2$ and $k_A=k_B=k$. It is easy to see that in this case, the output of the circuit in Fig.~\ref{figapp:schmidt-patch-circuit} leads to the final state (define $n'\equiv n-k-r$):
    \begin{align}   
        \ket{\psi_{M,R}}&=2^{-r/2-k} \sum_{\substack{x \in\{0,1\}^r \\ a, b \in\{0,1\}^k}}\ket{\sigma_A\left(x.a.0^{n'}\right)}_A\ket{\sigma_B\left(x.b.0^{n'}\right)}_B\,\\
        &=\dfrac{1}{M\sqrt{R}} \sum_{j=1}^{R} \sum_{x\in A_j,~y\in B_j}\ket{x}_A\ket{y}_B.
    \end{align}
By adding the pseudorandom phase function we obtain the state $\ket{\psi_{M,R,f}}$ in Eq.~\eqref{eq:schmidt-patch}. Upon inspection, it is easy to notice that if $\sigma_{A}, \sigma_{B}$ were genuine random samples from $S_{2^{n/2}}$, the above results in a uniform random sample from the corresponding set of patched subset phase states.

\subsection{Schmidt-patch states are close to Haar random states}\label{app:SchmidtPatchHaarProof}

In this section, we prove that both the Schmidt-patch subset state and the Schmidt-patch subset-phase states are statistically close to Haar random states. For simplicity, we consider $M_A=M_B=M$

\begin{lemma}[Proximity to Haar]\label{lem:HPSShaar}
    Consider the construction \eqref{eq:patchedsubsetstate} of $(M,R,N_A, N_B)$-Schmidt-patch subset states across an $A-B$ cut, such that subsystem $A$ is of $\ell$ qubits while subsystem $B$ is of $n-\ell$ qubits. In other words, $N_A=2^{\ell}$ and $N_B=2^{n-\ell}$, with $N=N_A\cdot N_B=2^n$. Also, define $D=\min\{N_A, N_B\}$. Furthermore, let the rank across the cut be $R=2^{\polylog(n)}$, and the patched disjoint subsets be of size $\abs{A_j}=\abs{B_j}=M$. Assume furthermore that $\frac{M^2R^{3/2}}{N}<1$. Then, random subset states of this form are close to Haar under $t=\poly(n)$ number of copies:
    \begin{equation}
    \label{eq:closeness-Haar-schmidt-subset}
\norm{\E_{A_{1:R},B_{1:R}}\dyad{\Psi_{A_{1:R},B_{1:R}}}^{\otimes t} - \E_{\phi\sim\mathrm{Haar}(\calH_A\otimes \calH_B)}\dyad{\phi}^{\otimes t}}_1\leq O\left(\frac{t^3}{M}\right) + O\left(\frac{t^2}{D}\right)+O\left(\frac{t}{\sqrt{R}}\right).
    \end{equation}
    As a corollary, introducing random global phases makes the associated subset-phase ensemble close to Haar irrespective of the value of $R$:
    \begin{equation}
        \label{eq:closeness-Haar-schmidt-subset-phase}
\norm{\E_{A_{1:R},B_{1:R},f}\dyad{\Psi_{A_{1:R},B_{1:R},f}}^{\otimes t} - \E_{\phi\sim\mathrm{Haar}(\calH_A\otimes \calH_B)}\dyad{\phi}^{\otimes t}}_1\leq O\left(\frac{t^3}{M}\right) + O\left(\frac{t^2}{D}\right).
    \end{equation}
\end{lemma}

We prove Lemma \ref{lem:HPSShaar} using a series of non-trivial steps. Note that this reduces to a counting problem, as outlined in Lemma \ref{lemma:subsetstatistics} (for Schmidt-patch subset phase states) and Lemma \ref{lemma:subsetstatistics2} (for Schmidt-patch subset states), respectively.  Let us first introduce some definitions which capture the distribution of the subsets:

\begin{definition}[Doubly-unique $t$-subsets]
Assume $n$ is even. The {\bf\em doubly-unique $t$-subsets} of $\B^n$, denoted by $\Sigma_{2,t} \subset \binom{\B^{n}}{t}$, are defined as the $t$-subsets for which both the first halves of the strings are unique, as well as the second halves of the strings:
\begin{equation}
    \Sigma_{2,t} \equiv 
    \left\{
    \begin{aligned}
        &\{ x_i.y_i \}_{i=1}^t \in 
        \binom{\B^{n/2}\times \B^{n/2}}{t} \; : \;\;\forall i \neq j,\;
        & x_i \neq x_j,\; y_i \neq y_j
    \end{aligned}
    \right\}
\end{equation}
\end{definition}

\begin{remark}
    A standard birthday asymptotic tells us that almost all $t$-subsets are doubly-unique:
    \begin{equation}\label{eq:sigma2dim}
        \abs{\Sigma_{2,t}} = \binom{N}{t}\left(1 - O\left(\frac{t^2}{\sqrt{N}}\right)\right)\,.
    \end{equation}
\end{remark}

\begin{figure}[h!]
    \centering
    \begin{tikzpicture}[scale=0.7]

  \fill[rounded corners=2, blue, opacity=0.2] (-0.3, 0.3) rectangle (1.3, -0.3) {};
  \fill[rounded corners=2, fill=magenta, opacity=0.2] (-0.3, 0.7) rectangle (2.3, 1.3) {};
  \fill[rounded corners=2, fill=magenta, opacity=0.2] (4-0.3, 0.7) rectangle (5.3, 1.3) {};
  \fill[rounded corners=2, blue, opacity=0.2] (6-0.3, -0.3) rectangle (6.3, 0.3) {};
  \draw[black, line width=1pt, bend left=30] (0,0) -- (1,1);
  \draw[black, line width=1pt, bend left=30] (0,0) -- (0,1);
  \draw[black, line width=1pt, bend left=30] (1,0) -- (1,1);
  \draw[black, line width=1pt, bend left=30] (1,0) -- (2,1);
  \draw[black, line width=1pt, bend left=30] (6,0) -- (4,1);
  \draw[black, line width=1pt, bend left=30] (6,0) -- (5,1);
  \foreach \x in {1,...,7}
    \foreach \y in {0,1} {
      \fill (\x-1,\y) circle (0.12);
      \node[black] at (-0.3+\x-1, -0.3+1.6*\y) {\small\bf \x};
    }
  \node[blue] at (0.5,-1) {$X_1$};
  \node[magenta] at (1,2) {$Y_1$};
  \node[blue] at (6,-1) {$X_2$};
  \node[magenta] at (4.5,2) {$Y_2$};

\end{tikzpicture}
    \caption{Example of the notions introduced in Definition \ref{def:bipartitecomponents}. A subset of $[\sqrt{N}]\times [\sqrt{N}]$ can be naturally viewed as a bipartite graph; in the above example, the subset is $T=\{(1,1), (1,2), (2,2), (2,3), (7,5), (7,6)\}$. According to Definition \ref{def:bipartitecomponents}, the bipartite components are $X_1=\{1,2\}, Y_1=\{1,2,3\}$, and $X_2=\{7\},Y_2=\{5,6\}$. The number of unique half-strings would be $u_T=8$.}
    \label{fig:example1}
\end{figure}
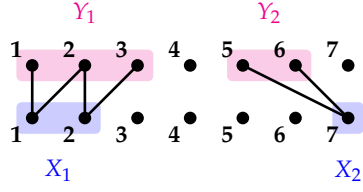

\begin{definition}[Bipartite components of subset]\label{def:bipartitecomponents}
    Consider an arbitrary $p$-subset of the Boolean cube $P=\{x_1.y_1,\dots,x_p.y_p\}\in\binom{\B^{n/2}\times \B^{n/2}}{p}$. We can interpret the elements $x_j.y_j$ as edges in a bipartite graph on $\B^{n/2}\times \B^{n/2}$. Consider the non-trivial (i.e. ignoring components which include a single vertex) connected components of this graph. If subset $P$ has $\ell$ such connected components, then each component is supported on a specific subset of the form $X_j\times Y_j \subset \B^{n/2}\times \B^{n/2}$. Let us denote the collection $(X_1,Y_1),\dots,(X_\ell, Y_\ell)$ the {\bf\em bipartite components} of subset $P$, which is uniquely defined for a specific $P$. Let us also denote:
    \begin{equation}
        u_P \equiv \abs{X_1}+\dots+\abs{X_\ell}+\abs{Y_1}+\dots+\abs{Y_\ell}
    \end{equation}
    as the number of {\bf\em total unique half-strings} in subset $P$.\\

    Formally: $(X_1,Y_1),\dots,(X_\ell, Y_\ell)$ are bipartite components of subset $P$ if $X_1,\dots,X_\ell$ are disjoint (not necessarily equal-sized) subset of $\B^{n/2}$ (and the same goes for $Y_1,\dots,Y_\ell$), and furthermore:
    \begin{equation}
\begin{aligned}
    & \forall\; j \in [\ell]: \\
    & \quad \forall\; x \in X_j,\quad \exists\; y \in Y_j\ \text{such that}\ x.y \in P,
    & \quad \nexists\; y \in \B^{n/2} \setminus Y_j\ \text{such that}\ x.y \in P, \\
    & \quad \forall\; y \in Y_j,\quad \exists\; x \in X_j\ \text{such that}\ x.y \in P,
    & \quad \nexists\; x \in \B^{n/2} \setminus X_j\ \text{such that}\ x.y \in P.
\end{aligned}
\end{equation}
See Figure \ref{fig:example1} for an illustration.
\end{definition}

\begin{definition}[Family of Schmidt-patch subset states]
    Let $\calF_{M,R,N} \in \binom{\B^{n/2}\times \B^{n/2}}{M^2R}$ be the family of $M^2R$-subsets of $[N]$ which patch $R$ pairs of disjoint $M$-subsets on the two half-cubes, as in construction \eqref{eq:patchedsubsetstate}.
\end{definition}

\begin{fact}\label{fact:counting1}
    Consider an arbitrary $p$-subset of the Boolean cube $P=\binom{\B^{n/2}\times \B^{n/2}}{p}$ with $t$ bipartite components and number of total unique half-strings $u_P$. Assuming $p < M$, we have that:
   \begin{equation}
   \label{eq:counting-ratio}
    \frac{\#\left\{ S \in \calF_{M,R,N} \,:\, P \subset S \right\}}{ \#\left\{ S \in \calF_{M,R,N} \right\} } 
    = \left( \frac{M}{\sqrt{N}} \right)^{u_P} R^t\times \left( 1 + O\left( \frac{p^3}{M} \right) + O\left( \frac{p^2}{\sqrt{N}} \right) \right).
\end{equation}
\end{fact}

\begin{proof}
    The $S$ subsets in the family are composed as:
    \begin{equation}
        S = A_1.B_1 \sqcup \dots \sqcup A_R.B_R\,,
    \end{equation}
    where $P.Q$ denotes all concatenations in the cartesian product, i.e. $P.Q = \{p.q\}_{p\in P,q\in Q}$\,. Here, $A_{1:R}$ are disjoint $M$-subsets of $[\sqrt{N}]$, and the same goes for $B_{1:R}$. We first count the number of ways $R$ pairs of boxes of size $M$ can be filled. The number of ways to fill the $A_{1:R}$ boxes with $M$ elements in each box is simply
    $$
    \binom{\sqrt{N}}{M} \binom{\sqrt{N}-M}{M}\cdots \binom{\sqrt{N}-RM}{M}=\dfrac{\sqrt{N}!}{M!^R(\sqrt{N}-MR)!}.
    $$
    The ordering of the $A_i$ boxes are not important and so the accounting for this fact we have that the number of (unordered) ways to fill the $A_{1:R}$ boxes is 
    $$
    \dfrac{\sqrt{N}!}{M!^R(\sqrt{N}-MR)!}\times \dfrac{1}{R!}\,.
    $$
    By the same argument, we can fill up the $B_{1:R}$ boxes with $M$ elements. So overall, the $A_{1:R}$ and $B_{1:R}$ boxes can be filled up in 
     $$
    \left(\dfrac{\sqrt{N}!}{M!^R(\sqrt{N}-MR)!}\times \dfrac{1}{R!}\right)^2
    $$
    ways.
    Finally, note that in order to count the number of $S$-subsets, we need to pair up the $A$ and $B$ boxes. This matching can be done in $R!$ ways. Thus the total number of $S$ subsets is:
    \begin{align}
    \label{eq:denominator-ratio}
        \# \{S \in \calF_{M,R,N}\} = \frac{1}{R!}\left(\frac{\sqrt{N} !}{M!^R(\sqrt{N}-MR)!}\right)^2\,.
    \end{align}

    Having obtained the denominator of the ratio in the LHS of Eq.~\eqref{eq:counting-ratio}, we now focus on the numerator.

    Let the bipartite set components making up the set $P$ be $(X_1,Y_1),\dots,(X_t,Y_t)$. In order to calculate the denominator, i.e.\ $\#\left\{ S \in \calF_{M,R,N} \right\}$, we need to enumerate the  ways in which the elements of $P$ can be allocated into the $A_j.B_j$ blocks of $S$. For this, we see that each bipartite component pair must belong to the same $A_j.B_j$ block, i.e.,\ $X_k \in A_j$ iff $Y_k \in B_j$. Therefore, we can restrict the problem to enumerating the number of ways in which the $t$ component pairs are allocated across the blocks. Let $\mu=1^{\mu_1}2^{\mu_2}\dots t^{\mu_t}$ be an integer partition of $t$. Denote by $\abs{\mu}=\mu_1+\dots+\mu_t$ the number of parts in $\mu$. For a specific integer partition $\mu$, consider all set partitions of $\{(X_1,Y_1),\dots,(X_t,Y_t)\}$ that follow the pattern of $\mu$, and consider allocations of these parts into the $A_j.B_j$ blocks that make up $S$. In other words, we would like to count all subsets $S$ such that $\mu_1$ blocks contain a single $(X_k,Y_k)$ component, $\mu_2$ blocks contain two components, etc. 

    In order to enumerate this quantity, let us first look at the number of ways to partition the set of $t$ components according to the partition type $\mu$. For this, notice the number of ways by which we can choose $\mu_1$ (the number of boxes holding a single component) out of the $t$ components is $\binom{t}{\mu_1}$. Then, we can choose $\mu_2$ (boxes holding two components) out of $\binom{t-\mu_1}{\mu_2}$ ways. Continuing in this way, we obtain that the total number of ways in which $t$ components can be distributed according to the partition $\mu$ is simply the multinomial coefficient
    $$
    \binom{t}{\mu_1,\mu_2,\cdots, \mu_t}=\equiv\dfrac{t!}{\mu_1!\mu_2!\dots \mu_t!}
    $$
    Furthermore, the $j$ components in each of the $\mu_j$ boxes need not be ordered. We need to account for this $1/j!^{\mu_j}$ ways. So overall, we can allocated $t$ components according to $\mu$ in
    $$
    \stirlingII{t}{\mu}=\dfrac{t!}{1!^{\mu_1}\dots t!^{\mu_t}\mu_1!\dots \mu_t!}.
    $$
 Here, we define $\stirlingII{t}{\mu}$ is the number of ways of partitioning the set of $t$ components according to partition type $\mu$, defined analogously to the Stirling number of the second kind ($\stirlingII{t}{k}$ is the Stirling number of the second kind, denoting the number of set partitions of $[t]$ into $k$ non-empty parts).

Next, for a particular partition $\mu$, we count the number of ways the remaining elements of the $A_{1:R}$ and $B_{1:R}$ boxes can be filled so that we make sure each box has exactly $M$ elements. Observe that we have $A_i\cdot B_i$ boxes that are already \textit{occupied} by the bipartite components. Given $\mu$, there are $|\mu|$ such \textit{occupied} pairs of boxes in all. The rest of the $R-|\mu|$ boxes are \textit{unoccupied}: they do not contain any of these components. Both these categories of boxes (\textit{occupied} and \textit{unoccupied}) need to be filled with $M$ elements each. This leads us to two cases.

Let us first consider the $A_{1:R}$ boxes that have been \textit{occupied} by the bipartite components. Each of the $A_i$ boxes need to be filled until they have $M$ elements each.  The number of ways the \textit{occupied} $A_{1:R}$ boxes can be filled up is given by
\begin{align*}
\binom{\sqrt{N}-\sum_{j=1}^{t}|X_j|}{M-\sum_{X_i\in \text{part }1}|X_i|,~M-\sum_{X_i\in \text{part }2}|X_i|,\cdots,~M-\sum_{X_i\in \text{part }|\mu|}|X_i|}
=
\dfrac{\left(\sqrt{N}-\sum_{j=1}^{t}|X_j|\right)!}{\prod_{j=1}^{|\mu|}\left(\sqrt{N}-\sum_{j=1}^{t}|X_j|\right)!\times \left(N-M|\mu|\right)!}.
\end{align*}
Similarly, in the same number of ways, we can fill up the \textit{occupied} $B_{1:R}$ boxes. So, the number of ways to fill up the \textit{occupied} $A_{1:R}B_{1:R}$ boxes is given by
\begin{equation}
\label{eq:occupied-boxes}
    \dfrac{\left(\sqrt{N}-\sum_{j=1}^{t}|X_j|\right)!\times \left(\sqrt{N}-\sum_{j=1}^{t}|Y_j|\right)!}{\prod_{j=1}^{|\mu|}\left(\sqrt{N}-\sum_{j=1}^{t}|X_j|\right)!\times \prod_{j=1}^{|\mu|}\left(\sqrt{N}-\sum_{j=1}^{t}|Y_j|\right)!\times \left(N-M|\mu|\right)!^2}.
\end{equation}

Next, we deal with the remaining $R-|\mu|$  boxes that did not contain any bipartite components, which we referred previously to as \textit{unoccupied} boxes. Since these need to be filled with exactly $M$ elements, the enumeration problem is similar to the arguments that went into the derivation of Eq.~\eqref{eq:denominator-ratio}. Let us deal with allocating strings to the \textit{unoccupied} $A$ boxes first. This can be done in
$$
\dfrac{\left(\sqrt{N}-M|\mu|\right)!}{M!^{R-|\mu|}\times \left(\sqrt{N}-MR\right)!\times (R-|\mu|)!},
$$
ways. The $B$ boxes can also be filled up in the same number of ways. Overall, the number of ways to allocate elements to these remaining $R-|\mu|$ boxes (also taking into account that these boxes can be matched up in $(R-|\mu|)!$ ways) is given by 
\begin{equation}
\label{eq:unoccupied-boxes}
    \dfrac{\left(\sqrt{N}-M|\mu|\right)!^2}{M!^{2(R-|\mu|)}\times (\sqrt{N}-MR)!^{2}\times (R-|\mu|)!}
\end{equation}
The product of Eq.~\eqref{eq:occupied-boxes} and Eq.~\eqref{eq:unoccupied-boxes} gives us the required number of ways $t$ components can be allocated to the pairs of boxes for a specific partition $\mu$. Then, overall, we have that:
\begin{align}
\label{eq:numerator-ratio}
    &\#\left\{ S \in \calF_{M,R,N}\,:\, P \subset S \right\} \nonumber \\
    &=\mathlarger{\mathlarger{\sum}}_{\mu} \stirlingII{t}{\mu}\times\dfrac{\left( \sqrt{N} - \sum_{i=1}^t|X_i|\right)!\times\left( \sqrt{N} - \sum_{i=1}^t|Y_i| \right)!} { (R - |\mu|)!\times M!^{2(R - |\mu|)}\times \left(\sqrt{N} - MR \right)!^2} \times \prod_{j=1}^{|\mu|}\dfrac{1}{( M - \sum\limits_{X_i \in \text{part } j}|X_i|)!}\times \prod_{j=1}^{|\mu|}\dfrac{1}{( M - \sum\limits_{Y_i \in \text{part } j}|Y_i|)!}.
\end{align}

    The above seems complicated but in the birthday regime it simplifies considerably. In particular, notice that the number of unique half-strings decouples from the way they are allocated across components inside $P$. For instance, while taking the ratio between Eq.~\eqref{eq:numerator-ratio} and Eq.~\eqref{eq:denominator-ratio}, we can use birthday approximations such as
$$
\dfrac{(\sqrt{N}-\sum_{j=1}^{t}|X_j|)!}{\sqrt{N}!}=\dfrac{1}{\left(\sqrt{N}\right)^{\sum_{j=1}^{t}|X_j|}}\left(1+O\left(\dfrac{\left(\sum_{j=1}^{t}|X_j|\right)^2}{\sqrt{N}}\right)\right)
$$
Also,
\begin{align*}
\prod_{j=1}^{|\mu|}\left(\dfrac{M!}{ (M-\sum_{X_i\in \text{part } j}|X_i|)!}\right)&=\prod_{j=1}^{|\mu|} M^{\sum_{X_i\in \text{part }j}|X_i|}\left(1+O\left(\dfrac{\left(\sum_{X_i\in \text{part }j}|X_i|\right)^2}{M}\right)\right)\\
            &=M^{\sum_{j=1}^{t}|X_j|}\left(1+O\left(\dfrac{\sum_{j=1}^{|\mu|} \left(\sum_{X_i\in \text{part }j}|X_i|\right)^2}{M}\right)\right)
\end{align*}

    When the smoke clears, we have that the desired ratio:
     \begin{align}
        \frac{\#\{S \in \calF_{M,R,N}\,:\,P\subset S\}}{\# \{S \in \calF_{M,R,N}\}} &= \left(\frac{M}{\sqrt{N}}\right)^{\abs{X_1}+\dots+\abs{X_t}+\abs{Y_1}+\dots+\abs{Y_t}}\left(1 + O\left(\frac{p^3}{M}\right)+O\left(\frac{p^2}{\sqrt{N}}\right)\right)\sum_{k=0}^t\stirlingII{t}{k}\frac{R!}{(R-k)!}\\
        &= \left(\frac{M}{\sqrt{N}}\right)^{u_P}R^t\left(1 + O\left(\frac{p^3}{M}\right)+O\left(\frac{p^2}{\sqrt{N}}\right)\right)\,.
    \end{align}
    where we used the standard combinatorial identity: $\sum_{k=0}^n\stirlingII{n}{k} x(x-1)\dots(x-k+1) = x^n$.
\end{proof}

This proof can be easily generalized to the asymmetric case where $A_{1:R}$ are disjoint $M$-subsets of $[2^{\ell}]$, while $B_{1:R}$ are disjoint $M$-subsets of $[2^{n-\ell}]$. This allows us to incorporate asymmetric Schmidt-patch states where subsystem $A$ is of $\ell$ qubits while $B$ is of $n-\ell$ qubits. So, $N_A=2^{\ell}$ and $N_B=2^{n-\ell}$. In this case consider (i)~$\calF_{M,R,N, \ell}\in \binom{\{0,1\}^{\ell}\times\{0,1\}^{n-\ell}}{M^2R}$, and (ii)~for $p<M$, $p$-subset of the Boolean cube $P=\binom{\{0,1\}^{\ell}\times\{0,1\}^{n-\ell}}{p}$. Then, it is easy to see that
   \begin{align}
        \# \{S \in \calF_{M,R,N, \ell}\} = \frac{1}{R!}\left(\frac{(2^{\ell}) !}{M!^R(2^{\ell}-MR)!}\right)\cdot \left(\frac{(2^{n-\ell}) !}{M!^R(2^{n-\ell}-MR)!}\right).
    \end{align}
On the other hand,
\begin{align}
    &\#\left\{ S \in \calF_{M,R,N}\,:\, P \subset S \right\} \nonumber \\
    &=\mathlarger{\mathlarger{\sum}}_{\mu} \stirlingII{t}{\mu}\times\dfrac{\left( N_A - \sum_{i=1}^t|X_i|\right)!\times\left( N_B - \sum_{i=1}^t|Y_i| \right)!} { (R - |\mu|)!\times M!^{2(R - |\mu|)}\times \left(N_A - MR \right)!\times\left(N_B - MR \right)!} \times \prod_{j=1}^{|\mu|}\dfrac{1}{( M - \sum\limits_{X_i \in \text{part } j}|X_i|)!}\times \prod_{j=1}^{|\mu|}\dfrac{1}{( M - \sum\limits_{Y_i \in \text{part } j}|Y_i|)!}.
\end{align}
Proceeding as before, the ratio of these two expressions simplify owing to birthday approximations. In the end, we have the following:
\begin{align}
        \frac{\#\{S \in \calF_{M,R,N,\ell}\,:\,P\subset S\}}{\# \{S \in \calF_{M,R,N,\ell}\}}
        &= \left(\dfrac{M^{u_p}}{N_A^{\sum_{i}|X_i|}\times N_B^{\sum_{i}|Y_i|}}\right)R^t\left(1 + O\left(\frac{p^3}{M}\right)+O\left(\frac{p^2}{2^{\min\{\ell,n-\ell\}}}\right)\right)\\
        &\leq \left(\dfrac{M}{D}\right)^{u_p}R^t\left(1 + O\left(\frac{p^3}{M}\right)+O\left(\frac{p^2}{D}\right)\right),
    \end{align}
where $D=\min\{N_A, N_B\}$.

\begin{fact}\label{fact:components}
    Consider two arbitrary $t$-subsets $T,T' \in \Sigma_{2,t}$ at distance $d(T,T')=d \geq 1$, where $d(T,T') \equiv \abs{T\cup T'}-t$ denotes the Hamming distance between two $t$-subsets. Denote by $u_{T\cup T'}$ the number of total unique half-strings in $T\cup T'$. Let the number of bipartite components in $T\cup T'$ be $\ell=t-d+k$. Then we have that:
    \begin{enumerate}
        \item[(a)] $1\leq k\leq 2d$.
        \item[(b)] $2t + \lceil\frac{4k-2d}{3}\rceil \leq u_{T\cup T'}\leq 2t + k$\,.
    \end{enumerate}
\end{fact}

\begin{proof}
    Let $T,T'\in \Sigma_{2,t}$ be two doubly unique $t$-subsets, such that $T=\{x_1.y_1, \dots, x_t.y_t\}$ and $T'=\{x_1'.y_1',\dots,x_t'.y_t'\}$. Denote disjoint sets $X_1,\dots,X_\l \subset \B^{n/2}$ and $Y_1,\dots,Y_\l\subset \B^{n/2}$ such that $(X_1,Y_1), \dots, (X_\l,Y_\l)$ form the bipartite components (cf. Definition \ref{def:bipartitecomponents}) of $T\cup T'$, i.e. the connected components of the corresponding bipartite graph with the elements of $T$ and $T'$ as the edges. Imagine coloring the $T$ and $T'$ vertices in two different colors:\\

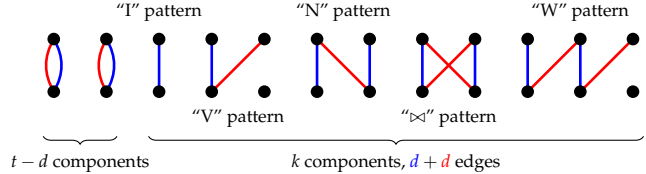
\begin{figure}[H]
    \centering
    \begin{tikzpicture}[scale=0.7]

      \draw[blue, line width=1pt] (0,0) to[out=60,in=-60] (0,1);
      \draw[red, line width=1pt] (0,0) to[out=120,in=240] (0,1);
      \draw[red, line width=1pt] (1,0) to[out=120,in=240] (1,1);
      \draw[blue, line width=1pt] (1,0) to[out=60,in=-60] (1,1);
      
      \draw[blue, line width=1pt] (2,0) to (2,1);
      
      \draw[blue, line width=1pt] (3,0) to (3,1);
      \draw[red, line width=1pt] (3,0) to (4,1);
      
      \draw[blue, line width=1pt] (5,0) to (5,1);
      \draw[red, line width=1pt] (5,1) to (6,0);
      \draw[blue, line width=1pt] (6,0) to (6,1);
      
      \draw[blue, line width=1pt] (7,0) to (7,1);
      \draw[red, line width=1pt] (7,0) to (8,1);
      \draw[blue, line width=1pt] (8,0) to (8,1);
      \draw[red, line width=1pt] (8,0) to (7,1);

      \draw[blue, line width=1pt] (9,0) to (9,1);
      \draw[red, line width=1pt] (9,0) to (10,1);
      \draw[blue, line width=1pt] (10,0) to (10,1);
      \draw[red, line width=1pt] (10,0) to (11,1);
      \foreach \x in {1,...,12}
        \foreach \y in {0,1} {
          \fill (\x-1,\y) circle (0.12);
        }
      \draw [decorate,decoration={brace,amplitude=3pt,mirror,raise=4ex}]
  (-0.2,0) -- (1.2,0) node[midway,yshift=-3em, scale=0.75]{$t-d$ components};
      \draw [decorate,decoration={brace,amplitude=3pt,mirror,raise=4ex}]
  (2-0.2,0) -- (11.2,0) node[midway,yshift=-3em, scale=0.75]{$k$ components, $\textcolor{blue}{d}+\textcolor{red}{d}$ edges};
      \node[scale=0.75] at (2,1.5) {``I'' pattern};
      \node[scale=0.75] at (3.5,-.5) {``V'' pattern};
      \node[scale=0.75] at (5.5,1.5) {``N'' pattern};
      \node[scale=0.75] at (7.5,-.5) {``$\bowtie$'' pattern};
      \node[scale=0.75] at (10,1.5) {``W'' pattern};
    
    \end{tikzpicture}
\end{figure}

The fact that $T,T'$ are doubly-unique $t$-subsets imposes some constraints. First, if the Hamming distance between the two $t$-subsets is $d(T,T')=d$, then there are $t-d$ elements in common, i.e. there are $t-d$ pairs $x_j.y_j \in T$ which are also pairs $x_i'.y_i' \in T'$. In terms of our graph picture, this means there are $t-d$ components $X_j.Y_j$ of size $\abs{X_j}=\abs{Y_j}=1$, represented by double-edged two-vertices sub-graphs. The other $k$ components must be made of non-doubled edges, since they belong to pairs $x_j.y_j \in T$ which do not exist also in $T'$. However, we must account for the combinatorics of the possible ways in which these mismatched elements can still share half-strings. The fact that $T,T'$ are doubly-unique subsets means that we cannot have two edges of the same color (i.e. belonging to the same subset) sharing a vertex.\\

Notice that the remaining $k$ bipartite components are made of a number of edges equal to $d$ edges coming from $T$ (depicted in blue), plus another $d$ edges belonging to $T'$ (denoted in red). This immediately implies the first claim: $1 \leq k \leq 2d$. The lower bound is achieved if all edges form one component, and the upper bound is achieved when all of the edges are in a component of their own.\\

Next, because of the double-uniqueness of $T$ and $T'$, these $k$ components can be of the following types:
\begin{enumerate}
    \item[(a)] Components containing only one edge (an ``I'' pattern), two edges sharing a vertex (a ``V'' pattern), or three edges in an ``N'' pattern. This means the number of vertices in a components is equal to the number of edges plus one.
    \item[(b)] If made of four or more edges, the edges can either form a closed loop (an ``bowtie'' $\bowtie$ pattern), or an open line (a ``W'' pattern). This means that the number of vertices in a component is either equal to the number of edges, or is equal to the number of edges plus one.
\end{enumerate}

When allocating $d$ $T$-edges and $d$ $T'$-edges across $k$ components, we must make use of at least $\max\left(0, \lceil\frac{4k-2d}{3}\rceil\right)$ components of size 1, 2, or 3. For such components, we know that the number of vertices is forced to be equal to the number of edges. As noted above, for components with four or more edges, the number of vertices can be equal to the number of edges, or equal to the number of edges plus one. Therefore, to count the number of unique half-strings $u_{T\cup T'}$ we must add $2(t-d)$ (coming from the components which correspond to $T\cap T'$) and the number of vertices coming from the above counting argument. This leads to the stated bound.

\end{proof}

\begin{fact}
    Consider two arbitrary $t$-subsets $T,T' \in \Sigma_{2,t}$, and assume $2t\leq M$ and also $\frac{M^2R^{3/2}}{N}<1$. Then we have that:
 \begin{equation}
\begin{aligned}
    \frac{\#\left\{\, S \in \calF_{M,R,N} \; : \; T \cup T' \subset S \,\right\}}{ \#\left\{\, S \in \calF_{M,R,N} \,\right\} }
    &=
    \begin{cases}
        \left( \dfrac{M^2 R}{N} \right)^t \\
        \quad \times \left( 1 + O\left( \dfrac{t^3}{M} \right) + O\left( \dfrac{t^2}{\sqrt{N}} \right) \right), \\
        \hfill \text{if } d(T, T') = 0, \\[2ex]
        \leq \left( \dfrac{M^2 R}{N} \right)^t \\
        \quad \times \left( \dfrac{1}{R} \right)^{d/2} \\
        \quad \times \left( 1 + O\left( \dfrac{t^3}{M} \right) + O\left( \dfrac{t^2}{\sqrt{N}} \right) \right), \\
        \hfill \text{if } d(T, T') = d \geq 1.
    \end{cases}
\end{aligned}
\end{equation}
\end{fact}

\begin{proof}
    Letting $(X_1,Y_1),\dots,(X_\l,Y_\l)$ be the bipartite components of $T\cup T'$, Fact \ref{fact:counting1} gives us that:
    \begin{equation}
        \frac{\#\left\{\, S \in \calF_{M,R,N} \; : \; T \cup T' \subset S \,\right\}}{ \#\left\{\, S \in \calF_{M,R,N} \,\right\} } = \left(\frac{M}{\sqrt{N}}\right)^{u_{T\cup T'}}R^\l\left(1 + O\left(\frac{t^3}{M}\right) + O\left(\frac{t^2}{\sqrt{N}}\right)\right).
    \end{equation}
    The case $d(T,T')=0$ is trivial since it implies $T=T'$, meaning the number of unique half-strings is $u_{T\cup T'}=2t$ and the number of bipartite components is $\l=t$.\\

    For the more generic case $d(T,T')\equiv d \geq 1$, let $k=\l+d-t$ where $\l$ once again denotes the number of bipartite components. Note that $1 \leq k \leq 2d$. Using Fact \ref{fact:components} we have that:
    \begin{align}
        1 \leq k \leq \frac{d}{2} &\quad \Rightarrow \quad u_{T\cup T'}\geq 2t \\
        \frac{d}{2}\leq k \leq 2d &\quad \Rightarrow \quad u_{T\cup T'}\geq 2t + \frac{4k-2d}{3}.
    \end{align}
    This implies that:
    \begin{equation}
    \begin{aligned}
        \left(\frac{M}{\sqrt{N}}\right)^{u_{T\cup T'}}R^\l &\leq\left\{
        \begin{array}{lcr}
            \left(\frac{M^2R}{N}\right)^{t}R^{-d+k} &\leq \left(\frac{M^2R}{N}\right)^{t}R^{-d/2} &\text{ if }1 \leq k \leq \frac{d}{2}\\
            \left(\frac{M^2R}{N}\right)^{t}\left(\frac{M^{2/3}R^{1/2}}{N^{1/3}}\right)^{2k}\left(\frac{M^{2/3}R}{N^{1/3}}\right)^{-d} &\leq \left(\frac{M^2R}{N}\right)^{t}R^{-d/2} &\;\text{ if }\frac{d}{2}\leq k \leq 2d \;\text{ and }\; \frac{M^2R^{3/2}}{N}<1.
        \end{array}\right.
    \end{aligned}
    \end{equation}
    Therefore, as long as $\frac{M^2R^{3/2}}{N}<1$, we have that in all cases we have the upper bound $\left(\frac{M}{\sqrt{N}}\right)^{u_{T\cup T'}}R^\l\leq \left(\frac{M^2R}{N}\right)^{t}R^{-d/2}$, which leads to the conclusion.
\end{proof}
So, from Lemma \ref{lemma:subsetstatistics2}, we obtain the following fact:
\begin{fact}
\label{fact:subset-state-cube}
    Assume $\frac{M^2R^{3/2}}{N}<1$. Then:
    \begin{equation}
        \norm{\E_{S\sim \calF_{M,R,N}}\dyad{S}^{\otimes t} - \E_{\phi\sim\mathrm{Haar}\left(\C^N\right)}\dyad{\phi}^{\otimes t}}_1 \leq O\left(\frac{t^3}{M}\right) + O\left(\frac{t^2}{\sqrt{N}}\right) + O\left(\frac{t}{\sqrt{R}}\right).
    \end{equation}
\end{fact}
~\\~\\
Similarly, from Lemma \ref{lemma:subsetstatistics}, we obtain the following.
~\\
\begin{fact}
\label{fact:subset-phase-state-cube}
    Assume $\frac{M^2R^{3/2}}{N}<1$. Then:
    \begin{equation}
        \norm{\E_{S\sim \calF_{M,R,N}, f\sim\{0,1\}^N}\dyad{\psi_{f,S}}^{\otimes t} - \E_{\phi\sim\mathrm{Haar}\left(\C^N\right)}\dyad{\phi}^{\otimes t}}_1 \leq O\left(\frac{t^3}{M}\right) + O\left(\frac{t^2}{\sqrt{N}}\right).
    \end{equation}
\end{fact}

Finally, we arrive at Eq.~\eqref{eq:closeness-Haar-schmidt-subset} and Eq.~\eqref{eq:closeness-Haar-schmidt-subset-phase} from Fact \ref{fact:subset-state-cube} and Fact \ref{fact:subset-phase-state-cube}, respectively. This completes the proof of Lemma \ref{lem:HPSShaar}.

\section{Large-brick phase states}\label{sec:LargeBrickProofs}
In this section, we formally prove our main results with respect to the properties of large-brick phase states. We begin with the proof of \cref{fact:LargeBrickGlobalPseudorandom}, which claims that these states are close to Haar random states.
\subsection{Proof of \cref{fact:LargeBrickGlobalPseudorandom}}
\label{subsec-app:lbps-close-to-haar}
Define the ``birthday-subspace'' (i.e. the collision-free subspace) in the symmetric subspace $\Sym{t}{\C^N}$ as:
\begin{equation}
    \Sigma_{\mathrm{bday}} = \left\{\frac{1}{\sqrt{t!}}\sum_{\sigma\in \mathbb{S}_t}R_\sigma\ket{x_1^{(1)}\cdots x_{2m}^{(1)}\;\;\cdots\;\;x_1^{(t)}\cdots x_{2m}^{(t)}}\;:\;\left(x_{j}^{(k)}\right)_{\substack{j\in[2m]\\k\in[t]}}\in \left(\{0,1\}^b\right)^{2mt} \text{all distinct}\right\}\,,
\end{equation}
where the symmetric group acts by permuting the $t$ sites, as usual:
\begin{equation}
    R_\sigma\ket{x_1^{(1)}\cdots x_{2m}^{(1)}\;\;\cdots\;\;x_1^{(t)}\cdots x_{2m}^{(t)}} = \ket{x_1^{(\sigma^{-1}(1))}\cdots x_{2m}^{(\sigma^{-1}(1))}\;\;\cdots\;\;x_1^{(\sigma^{-1}(t))}\cdots x_{2m}^{(\sigma^{-1}(t))}}.
\end{equation}
Note that the birthday subspace has dimension:
\begin{align*}
    \mathrm{dim}\;\Sigma_{\mathrm{bday}} &= \frac{2^b(2^b-1)\cdots (2^b-2mt+1)}{t!}\\
    &= \frac{2^{nt}}{t!}\left(1-O\left(\frac{m^2t^2}{2^b}\right)\right)\\
    &= \mathrm{dim}\;\Sym{t}{\C^N}\left(1-O\left(\frac{n^2t^2}{2^b}\right)\right)\,,\label{eq:DimBdaySubspace}
\end{align*}
since $2bm=n$ and $m<n$, and using the fact that $\mathrm{dim}\;\Sym{t}{\C^N}=\binom{2^n+t-1}{t}=\frac{2^{nt}}{t!}\left(1+O\left(\frac{t^2}{2^n}\right)\right)$.

Let $\Pi_{\mathrm{bday}}$ be the projector on the above birthday subspace $\Sigma_{\mathrm{bday}}$. Let the two density matrices of interest be:
\begin{align}
    \rho_{\mathrm{Haar}}^{(t)} &= \E_{\ket{\phi} \sim \mathrm{Haar}(\C^N)} \dyad{\phi}^{\otimes t} \\
    \rho_{\mathrm{LBPS}}^{(t)} &= \E_{f_{1:m},g_{1:m}}\dyad{\psi_{f_{1:m},g_{1:m}}}^{\otimes t}.
\end{align}
By their manifest symmetry, both density matrices belong to the symmetric subspace $\Sym{t}{\C^N}$. Moreover,
\begin{equation}\label{eq:HaarMomentProjector}
     \rho_{\mathrm{Haar}}^{(t)} = \frac{\Pi_{\Sym{t}{\C^N}}}{\mathrm{dim}\,\Sym{t}{\C^N}}.
\end{equation}
We will show that the two density matrices are close on the birthday subspace. Begin by evaluating the large-block phase-state ensemble explicitly:
\begin{align}
    \rho_{\mathrm{LBPS}}^{(t)} &= \frac{1}{N^t}\sum_{\substack{\left(x_j^{(k)}\right)_{j\in[2m],k\in[t]} \in \left(\B^b\right)^{2mt} \\ \left(y_j^{(k)}\right)_{j\in[2m],k\in[t]} \in \left(\B^b\right)^{2mt}}} \E_{f_{1:m},g_{1:m}} (-1)^{\varphi(x,y)} \dyad{x_1^{(1)}\cdots x_{2m}^{(1)}\;\;\cdots\;\;x_1^{(t)}\cdots x_{2m}^{(t)}}{y_1^{(1)}\cdots y_{2m}^{(1)}\;\;\cdots\;\;y_1^{(t)}\cdots y_{2m}^{(t)}}\,,
\end{align}
where the phase function is:
\begin{equation}
    \varphi(x,y) = \sum_{\substack{j\in[m]\\k\in[t]}} f_j(x_{2j-1}^{(k)}.x_{2j}^{(k)}) + g_j(x_{2j-2}^{(k)}.x_{2j-1}^{(k)}) + f_j(y_{2j-1}^{(k)}.y_{2j}^{(k)}) + g_j(y_{2j-2}^{(k)}.y_{2j-1}^{(k)})\,.
\end{equation}
Restricting from $\rho_{\mathrm{LBPS}}^{(t)}$ to the birthday subspace $\Pi_{\mathrm{bday}}\rho_{\mathrm{LBPS}}^{(t)}\Pi_{\mathrm{bday}}$ makes it easier to evaluate the phase averages under random functions $f_j$ and $g_j$. Notice that the average over $f_j$ is zero unless the same argument of $f_j$ appears an even number of times in the sum $\varphi(x,y)$. Since on the birthday subspace any sub-string $x_j^{(k)}$ of length $b$ can only occur once, so the above average over $f_1,\dots,f_{m}$ is one if and only if for every pair $x_{2j-1}^{(a)}.x_{2j}^{(a)}$ there is a single equal pair $y_{2j-1}^{(b)}.y_{2j}^{(b)}$ (matchings with more than two equivalent terms are also forbidden by the unique sub-string condition on $\Sigma_{\mathrm{bday}}$). Similarly, averaging over the functions $g_1,\dots,g_{m}$ is one if and only if for any pair $x_{2j-2}^{(a)}.x_{2j-1}^{(a)}$ there exists a single equal pair $y_{2j-2}^{(b)}.y_{2j-1}^{(b)}$.

Combining the two constraints only leaves configurations such that for each $n$-bit site $x_1^{(a)}\cdots x_{2m}^{(a)}$ there exists a single identical $n$-bit site $y_1^{(b)}\cdots y_{2m}^{(b)}$. In other words, only $t$-permutations survive:
\begin{align}
    \Pi_{\mathrm{bday}}\rho_{\mathrm{LBPS}}^{(t)}\Pi_{\mathrm{bday}} &= \frac{1}{N^t}\sum_{\substack{(x^{(1)},\dots,x^{(t)})\in\Sigma_{\mathrm{bday}}\\\sigma\in\mathbb{S}_t}}\dyad{x^{(1)},\dots,x^{(t)}}{y^{(\sigma(1))},\dots,y^{(\sigma(t))}} \\
    &= \frac{t!}{N^t}\sum_{\mathbf{x}\in\Sigma_{\mathrm{bday}}}\dyad{\mathbf{x}}\\
    &= \frac{1 + O\left(\frac{t^2}{2^n}\right)}{\binom{2^n+t-1}{t}}\Pi_{\mathrm{bday}}\,.
\end{align}
Using \eqref{eq:HaarMomentProjector} and \eqref{eq:DimBdaySubspace}, we have that the two density matrices are close on the birthday subspace:
\begin{equation}
    \left\|\Pi_{\mathrm{bday}}\rho_{\mathrm{LBPS}}^{(t)}\Pi_{\mathrm{bday}} - \Pi_{\mathrm{bday}}\rho_{\mathrm{Haar}}^{(t)}\Pi_{\mathrm{bday}}\right\|_1 = O\left(\frac{t^2n^2}{2^b}\right).
\end{equation}
Using Lemma 1 from \cite{giurgicatiron2023pseudorandomnesssubsetstates}, we can lift this to the full symmetric subspace to get the desired result:
\begin{equation}
\left\|\rho_{\mathrm{LBPS}}^{(t)} - \rho_{\mathrm{Haar}}^{(t)}\right\|_1 = O\left(\frac{t^2n^2}{2^b}\right).
\end{equation}
Pseudorandomness follows from replacing the random functions with quantum-secure pseudorandom functions (which admit efficient circuits). The circuit architecture immediately implies a bound of $O(b)$ for the entanglement entropy across geometric cuts, implying that this ensemble is pseudoentangled along a 1D geometry.\qed 

\subsection{Large-brick phase states form projected designs across almost all cuts}
\label{subsec-app:lbps-projected-design}
Here we prove that for the large-brick phase state,  by measuring a random biparition of qubits results in a projected ensemble that is pseudoentangled. That is, these states lead to a pseudoentangled projected design when a measurement is made across almost all cuts.

Given a $n$-bit string $x\in\B^n$, partition it into substrings corresponding to the bricks of the two-layer brickwork circuit preparing the phase states: $x=x_1.x_2.\cdots.x_{2m}$, where $x_1,\cdots,x_{2m} \in \B^b$. Then, any cut further partitions these bricks into segments which belong to one side or the other of the cut: $x_j=a_j.b_j$ where $a_j\in\B^{k_j}$ and $b_j\in\B^{b-k_j}$, for $j\in[2m]$ and some $k_1,\dots,k_{2m}\in\{0,\dots,b\}$.

We can treat a random cut (i.e. bipartition) of $n$ qubits as $n$ i.i.d. draws from a $\Ber(1/2)$ distribution, with the $j$'th draw determining whether qubit $j$ is on one side or the other of the cut. Then each $k_j$ follows the standard binomial statistics, in particular with a tail bound:
\begin{equation}
    \Pr[k_j\leq k]\leq \exp[-\frac{(b-2k)^2}{4b}].
\end{equation}
A similar tail applies for the upper tail of this quantity. Then, by a union bound over these $2m$ events, we know that the typical cut segmentation cannot be too uneven on even a single brick:
\begin{equation}\label{eq:BinomialConcentration}
    \Pr\left[\exists j\in[2m]\,:\,k_j\notin \left[\frac{b}{4},\frac{3b}{4}\right]\right] \leq 4ne^{-b/16}.
\end{equation}
This is negligible in $n$ if $b\geq\Omega(\log^2n)$, as assumed in our construction.

To study the projected ensemble properties across a random cut of this kind, assume we are measuring in the computational basis on the ``b'' side of the cut with measurement outcome $b=(b_1,\dots,b_{2m})\in\B^{n - \sum_{j\in[2m]}k_j}$. All outcomes occur with equal probability $2^{-n+\sum_{j\in[2m]}k_j}$. The corresponding post-measurement state on the ``a'' side of the cut is:
\begin{equation}
    \ket{\psi_{f_{1:m},g_{1:m}}^{a|b}} = \frac{1}{2^{\sum_{j\in[2m]}k_j/2}}\sum_{\substack{a_1\in\B^{k_1}\\\vdots\\a_{2m}\in\B^{k_{2m}}}}(-1)^{\sum_{j\in[m]}f_j(a_{2j-1}.b_{2j-1}.a_{2j}.b_{2j})+g_j(a_{2j-2}.b_{2j-2}.a_{2j-1}.b_{2j-1})}\ket{a_1,\cdots,a_{2m}}.
\end{equation}
Notice that we can write this state as a two-layer brickwork phase-state on the ``a'' subsystem:
\begin{equation}
    \ket{\psi_{\tilde{f}_{1:m}^{(b)},\tilde{g}_{1:m}^{(b)}}^{a}} = \frac{1}{2^{\sum_{j\in[2m]}k_j/2}}\sum_{\substack{a_1\in\B^{k_1}\\\vdots\\a_{2m}\in\B^{k_{2m}}}}(-1)^{\sum_{j\in[m]}\tilde{f}^{(b)}_j(a_{2j-1}.a_{2j})+\tilde{g}^{(b)}_j(a_{2j-2}.a_{2j-1})}\ket{a_1,\cdots,a_{2m}},
\end{equation}
where we defined the $b$-outcome-conditioned ``post-measurement functions'' $\tilde{f}_j^{(b)}(a_{2j-1}.a_{2j})\equiv f_j(a_{2j-1}.b_{2j-1}.a_{2j}.b_{2j})$ and similarly for $\tilde{g}_j^{(b)}$. This resulting state is a ``large-brick'' phase-state if the bricks are large enough, e.g., $k_1,\dots,k_{2m}\geq \Omega(\log^2 n)$, which holds with high probability by the concentration result of \eqref{eq:BinomialConcentration}.

\section{Towards Pseudorandomness beyond BQP}
\label{sec:pseudorandomness-beyond-bqp}
\noindent In this section, we want to motivate pseudorandomness for distinguishers beyond $\mathsf{BQP}$. We will consider two different ensembles of states $\mathcal{B}_1$ and $\mathcal{B}_2$. The first ensemble $\mathcal{B}_1$ is an ensemble of pseudorandom phase states, given by
\[
\frac{1}{\sqrt{2^n}} \sum_{x \in \{0, 1\}^n} (-1)^{f(x)} \ket{x},
\]
where $f$ is a pseudorandom function. The second ensemble $\mathcal{B}_2$ is an ensemble of large brick phase states, as in Figure \ref{fig:largebricks1}, where the phases are picked pseudorandomly. The details of these pseudorandom functions will be explained later.

Now, let $\mathcal{D}$ be a polynomial time quantum distinguisher, given $t$ copies of a state $|\psi\rangle$---which is promised to be either from $\mathcal{B}_1$ (or $\mathcal{B}_2$) or a Haar random state---with the goal of $\mathcal{D}$ being to distinguish which is the case. $\mathcal{D}$ has access to a oracle $\mathcal{Q}_P$, which it can query only once. $\mathcal{D}$ interfaces with the oracle as follows.
\begin{itemize}
\item The input is the state $\ket{\psi}$ and a classical string describing
\[
C \;=\; \bigl(U_{1},M_{1},\dots,U_{P},M_{P}\bigr),
\]
where $U_i$ is a unitary operator on $n$ qubits, and $M_i$ is a \emph{collapsing} measurement operator on $s$ qubits with
  $0 \le s \le n$. 

Define the intermediate states
\[
\ket{\phi_{0}} \;=\; |\psi\rangle,
\qquad
\ket{\phi_{i}} \;=\; \frac{M_i U_i \ket{\phi_{i-1}}}{\|M_i U_i \ket{\phi_{i-1}}\|}.
\]

\item The oracle $\mathcal{Q}_{P}$, on input $C$, outputs the sequence
$\bigl\{\ket{\phi_{i}}\bigr\}_{i=0}^{P}$.
\end{itemize}

\noindent The setting just described is, more or less, captured by the class $\mathsf{PDQP}$, which is one of the classes ''just above`` $\mathsf{BQP}$, see for e.g. \cite{aaronson2016spacejustabovebqp, miloschewsky2025newlowerboundsquantumcomputation}. 

To instantiate the states in $\mathcal{B}_1$ and $\mathcal{B}_2$ efficiently, we conjecture the existence of $\mathsf{PDQP}$-secure pseudorandom functions. We leave open whether such functions exist. Note that $\mathsf{PDQP}$ contains $\mathsf{SZK}$ where a variety of lattice problems, otherwise hard for $\mathsf{BQP}$, become easy \cite{PeikertVaikuntanathan2008}. Hence, certain constructions based on the hardness of lattice problems, like Learning With Errors (LWE), may not be good candidates. However, other hardness conjectures, like Learning Parity with Noise (LPN) \cite{blum2000noisetolerantlearningparityproblem} can still potentially work, although a rigorous security proof is open and we leave it for future work.

Now, let us restrict the power of the distinguisher $\mathcal{D}$ and that of the oracle $\mathcal{Q}_P$. Firstly, let us take each $U_i$ to be identity and let us take each $M_i$ to be a (randomly sampled) standard basis projector on any cut. 

\begin{lemma}[Pseudorandomness beyond \textsf{BQP}]
Consider an ensemble $\mathcal{B}_1$ of pseudorandom phase states, given by
\[
\frac{1}{\sqrt{2^n}} \sum_{x \in \{0, 1\}^n} (-1)^{f(x)} \ket{x},
\]
\noindent where $f$ is a \textsf{PDQP}-secure pseudorandom function. Then, for every distinguisher $\mathcal{D}$, which can query an oracle $\mathcal{Q}_P$ of the form just described, we have that with high probability over the choice of the input state,
\begin{equation}
\left|\underset{{\psi} \sim \mathcal{B}_1}{\mathsf{Pr}}[\mathcal{D}(\ket{\psi}^{\otimes t}) = 1] -    \underset{{\psi} \sim \mathsf{Haar}}{\mathsf{Pr}}[\mathcal{D}(\ket{\psi}^{\otimes t}) = 1]     \right|
 = \negl(n),
\end{equation}
for all $t = \mathsf{poly}(n)$.
\end{lemma}
\noindent \begin{proof}
Random phase states are statistically pseudorandom, with exponentially close trace distance to a Haar random state; see for e.g. \cite{ananth2023pseudorandomfunctionlikequantumstate, aaronson_et_al:LIPIcs.ITCS.2024.2}. By a similar calculation to \Cref{subsec-app:lbps-projected-design}, it is easy to see that across any cut, the projected ensemble corresponding to a pseudorandom phase state is \emph{also} a pseudorandom phase state of smaller size. Hence, in the output sequence produced by the oracle $\mathcal{Q}_P$,
\[
\bigl\{\ket{\phi_{i}}\bigr\}_{i=0}^{P},
\]
every $\ket{\phi_{i}}$ is a pseudorandom state with high probability. The proof now follows from a series of hybrid arguments.
\end{proof}

Now, let us take each $U_i$ to be identity, $M_1$ to be a (randomly sampled) standard basis projector on a \emph{randomly chosen} cut, and all other collapsing measurements to be empty.

\begin{lemma}[Pseudoentanglement beyond \textsf{BQP}]
Consider an ensemble $\mathcal{B}$ of large-brick phase states as described in Figure \ref{fig:largebricks1}, which have entanglement scaling as $\sim \log^2 n$ across geometric cuts. Then, for every distinguisher $\mathcal{D}$, which can query an oracle $\mathcal{Q}_P$ of the form just described, we have that with high probability over the choice of the input state,
\begin{equation}
\left|\underset{{\psi} \sim \mathcal{B}_2}{\mathsf{Pr}}[\mathcal{D}(\ket{\psi}^{\otimes t}) = 1] -    \underset{{\psi} \sim \mathsf{Haar}}{\mathsf{Pr}}[\mathcal{D}(\ket{\psi}^{\otimes t}) = 1]     \right|
 = \negl(n),
\end{equation}
for all $t = \mathsf{poly}(n)$.
\end{lemma}
\noindent \begin{proof}
Since all but the first collapsing measurement is empty, the output sequence of $\mathcal{Q}_P$ is now polynomially many copies of a state from a projected ensemble which is pseudorandom with high probability, according to \Cref{fact:LargeBrickGlobalPseudorandom} and the argument in \Cref{subsec-app:lbps-projected-design}. The proof then follows from standard hybrid arguments.
\end{proof}
\end{document}